\documentclass[letterpaper, 10 pt, conference]{ieeeconf}
\IEEEoverridecommandlockouts                              
\usepackage{times}
\usepackage{multicol}
\usepackage{cite}
\usepackage{amsmath,amssymb,amsfonts}
\usepackage{bm}
\usepackage{graphicx}
\usepackage{tikz}
\usetikzlibrary{arrows.meta,positioning,matrix}
\usepackage{pgfplots}
\pgfplotsset{compat=1.16}
\let\labelindent\relax
\usepackage{enumitem}
\usepackage{algorithm}
\usepackage{algpseudocode}

\newtheorem{definition}{Definition}
\newtheorem{lemma}{Lemma}
\newtheorem{proposition}{Proposition}
\newtheorem{theorem}{Theorem}
\newtheorem{problem}{Problem}
\newtheorem{remark}{Remark}

\newtheorem{assumption}{Assumption}

\usepackage[export]{adjustbox}
\usepackage{float}
\usepackage{xcolor}
\usepackage{url}

\DeclareMathOperator*{\argmin}{arg\,min}
\DeclareMathOperator*{\argmax}{arg\,max}

\newcommand{\reals}{\mathbb{R}}
\newcommand{\nnreals}{\mathbb{R}_{\geq 0}}
\newcommand{\naturals}{\mathbb{N}}

\newcommand{\rulebook}{\mathcal{R}}
\newcommand{\ruleset}{R}
\newcommand{\rbrule}{r}
\newcommand{\realizations}{\Xi}
\newcommand{\preorder}{\lesssim}

\newcommand{\outcomes}{\Omega}
\newcommand{\algebra}{\Sigma}
\newcommand{\probability}{\mathit{Pr}}

\newcommand{\CVaR}{\text{CVaR}}

\newcommand{\risk}{\rho}

\newcommand{\threshold}{\gamma}

\newcommand{\rrrule}{e}

\newcommand{\policies}{\Pi}
\newcommand{\policy}{\pi}

\newcommand{\rbleq}{\lesssim_{\rulebook}}

\newcommand{\rblt}{<_{\rulebook}}

\newcommand{\radius}{\operatorname{rad}}

\title{
  \LARGE \bf
  Risk-Aware Optimal Control with Rulebooks
}

\author{
  Tichakorn Wongpiromsarn$^{1}$
  \thanks{*This work was supported in part by NSF under Grant
    CNS-2141153.
    The author gratefully acknowledges Konstantin Slutsky for valuable discussions and
    guidance that strengthened the mathematical rigor of this work.}
  \thanks{$^{1}$Tichakorn Wongpiromsarn is with the Department of Computer Science,
    Iowa State University, Ames, IA 50010
    {\tt\small nok@iastate.edu}
  }%
}

\begin{document}

\maketitle
\thispagestyle{empty}
\pagestyle{empty}

\begin{abstract}
  We consider safety-critical control problems involving multiple requirements with different
  priorities and uncertainty in their evaluation.
  We represent these requirements using risk-aware rulebooks,
  where each requirement is assigned a risk measure and an acceptable threshold,
  and a priority relation is defined among the requirements.
  Each requirement induces a risk-evaluation function that maps a policy to the risk associated with
  its violation.
  We formulate risk-aware optimal control with rulebooks as a lexicographic optimization problem
  over excess risks and develop an anytime filtering and branch-and-bound algorithm that
  progressively tightens the certified optimality gap while characterizing the
  corresponding set of policies at each priority level.
  The algorithm returns a policy together with these gaps, which bound its suboptimality.
  We prove that these gaps are valid for any finite computational budget and,
  under additional assumptions, converge to zero as the computational budget increases.
  We evaluate the algorithm on a synthetic benchmark with a known optimum and a
  realistic highway-merging simulation with CVaR-based collision, rear-braking, headway, and comfort
  rules.
\end{abstract}

\section{Introduction}
Control of safety-critical autonomous systems typically requires balancing several objectives,
including safety, regulatory requirements, progress, and other performance objectives,
which may have different priorities.
For example, in a highway merging maneuver, avoiding collision should have higher priority than obeying a
three-second headway rule, while headway in front and behind the ego vehicle may have the same
priority.
If all objectives cannot be simultaneously satisfied,
the controller should violate lower-priority objectives before higher-priority ones.
This creates a multi-objective optimization problem with a priority structure, not just a scalar cost
minimization problem.

Uncertain environments (e.g., interactions with other road users) further complicate
this problem because the same policy can produce different outcomes,
making its performance uncertain.
A standard formulation in stochastic control and reinforcement learning is to optimize the expected
cumulative cost or reward associated with a sequence of decisions
\cite{Bertsekas:2019:Reinforcement,Sutton:2018:Reinforcement},
but expectation can obscure the risk posed by rare outcomes with severe consequences.
Other risk measures capture different aspects of uncertainty.
Worst-case risk considers the most adverse possible outcome,
Value at Risk (VaR) specifies a quantile of the distribution of a quantity of interest, and
Conditional Value at Risk (CVaR) measures the expected risk in the tail beyond the VaR threshold
\cite{Majumdar:2020:How}.
CVaR is a standard example of a coherent tail-risk measure and admits a useful optimization
representation
\cite{Artzner:1999:Coherent,Rockafellar:2000:Optimization,Pflug:2000:Some}.
Risk-sensitive and risk-constrained formulations have been studied in MDPs
\cite{Chow:2015:RiskSensitive},
stochastic shortest-path planning \cite{Ahmadi:2021:RiskAverse},
reinforcement learning \cite{Chow:2018:RiskConstrained},
stochastic system verification \cite{Lindemann:2023:Risk},
and risk-aware motion planning and control
\cite{Hakobyan:2019:Risk,Yang:2026:Risk,Nyberg:2021:Risk}.
These works evaluate, optimize, or constrain risk within a specified model, but they do not directly
address policy selection under multiple prioritized risk criteria.

Rulebooks provide a formal way to describe priorities among requirements and objectives,
with each represented by a rule
\cite{Censi:2019:Liability,Wongpiromsarn:2026:Formal}.
Risk-aware rulebooks extend this formalism to uncertain environments by assigning a risk measure and a
threshold to each rule \cite{Wongpiromsarn:2026:Risk}.
Different rules can use different risk measures, including expected value,
worst-case risk, VaR, CVaR, or other measures.
Policies with risk below the threshold of a rule are considered equally acceptable for that
rule.
Thus, once a high-priority rule is certified to be below its threshold,
further reducing its risk provides no additional benefit,
and the algorithm can proceed to the next rule level.

This threshold-based priority structure leads to a different computational problem from standard
risk-sensitive or risk-constrained optimization.
A particularly related MDP planning method first optimizes CVaR and then, among all CVaR-optimal
policies, optimizes expected value \cite{Rigter:2022:Planning}.
Our formulation extends this idea by
allowing more than two rule levels,
allowing each level to use its own risk measure and threshold,
and applying when the risk-evaluation functions are black boxes rather than being derived from a fixed
MDP model.
Existing lexicographic optimization methods
\cite{Isermann:1982:Linear,Miettinen:1999:Nonlinear,Ehrgott:2005:Multicriteria},
lexicographic MDPs \cite{Wray:2015:Multi-Objective}, and lexicographic MPC methods
\cite{Rasekhipour:2018:Lexicographic} handle prioritized objectives and constraints.
However, they rely on problem-specific structure, such as explicit objective and constraint
functions, transition and reward models, or system dynamics and stage costs.
References \cite{Sui:2015:Safe,Turchetta:2019:Safe} also consider optimization and exploration
of unknown functions under Lipschitz assumptions,
but focus on safe exploration rather than prioritized optimization with rule-specific risk measures
and thresholds.

This paper develops an algorithm for computing an optimal policy under risk-aware rulebooks when the
risk-evaluation functions are black boxes.
We assume that each risk-evaluation function is Lipschitz continuous but
do not require gradients, convexity, or an analytic expression.
Our algorithm builds on the branch-and-bound principle of maintaining lower and upper
bounds over regions of the search space
\cite{Piyavskii:1972:Algorithm,Shubert:1972:Sequential,Horst:1996:Global}.
Unlike classical branch-and-bound methods,
which use these bounds to guide refinement toward a global minimizer of a single objective,
our algorithm maintains bounds at each rule level to eliminate regions
that cannot contain an optimizer with respect to the rule levels processed so far.
Thus, our refinement strategy 
focuses on characterizing the entire set of policies that remain potentially optimal and
providing optimality-gap guarantees under a finite computational budget,
rather than identifying a single global optimizer as in classical branch-and-bound.


The contributions of this paper are as follows.
First, we formulate risk-aware optimal control with rulebooks as a lexicographic optimization
problem over excess risks.
This extends existing risk-aware control formulations by allowing each rule to have its own risk
measure and threshold together with explicit priorities among rules.
Second, we introduce an anytime filtering and branch-and-bound algorithm for black-box
risk-evaluation functions.
Unlike conventional sequential approaches to lexicographic optimization,
our algorithm does not require each rule level to be solved exactly before proceeding to the
next.
Instead, it uses certified bounds obtained within a finite computational budget to
propagate certified policy sets across the rule levels.
Third, we prove that the certified optimality gap returned by the algorithm is valid for any finite
computational budget and, under additional assumptions, converges to zero as the computational budget
increases.
Fourth, we demonstrate the method on a synthetic benchmark with a known optimum and
on a realistic highway merging simulation using CVaR as the risk measure for collision,
rear-braking, headway, and comfort rules.

\section{Preliminaries}
\label{sec:prelim}

Let $\realizations$ denote the set of objects to be evaluated, e.g.,
trajectories that describe the states of all relevant agents in the world.
A rulebook compares these objects by quantifying their violations of each rule
and ordering the rules by priority \cite{Censi:2019:Liability,Wongpiromsarn:2026:Formal}.
Throughout the paper, $\reals$ and $\nnreals$ denote the set of real and non-negative real
numbers, respectively and $[x]_+ = \max\{x,0\}$ for any $x \in \reals$.

\begin{definition}
  \label{def:rule}
  A \emph{rule} is a function $\rbrule : \realizations \to \mathbb{R}_{\geq 0}$ that measures the
  degree of violation of its argument.
\end{definition}

$\rbrule(x) < \rbrule(y)$ means that  $x$ violates $\rbrule$ less than $y$ does.
In particular, $\rbrule(x) = 0$ indicates that  $x$ satisfies the rule.
Although we call $\rbrule$ a ``rule'', it can encode not only a legal or hard constraint
but also a performance criterion, such as comfort, efficiency, or user preference.

\begin{definition}
  \label{def:rulebook}
  A rulebook is a tuple $\rulebook = \langle \ruleset, \preorder \rangle$,
  where $\ruleset$ is a set of rules and $\preorder$ is a preorder on $\ruleset$ specifying their
  relative priority.
\end{definition}

A preorder generalizes a partial order by allowing distinct elements $s_{1}$ and $s_{2}$ to
satisfy both $s_1 \preorder s_2$ and $s_2 \preorder s_1$.
Throughout the paper, we assume that $\ruleset$ is finite.
For $\rbrule,\rbrule' \in \ruleset$, we write $\rbrule > \rbrule'$
to denote that $\rbrule$ has strictly higher priority than $\rbrule'$, i.e.,
$\rbrule \not\preorder \rbrule'$ and $\rbrule' \preorder \rbrule$.
The preorder allows three cases:
(1) strict priority ($\rbrule > \rbrule'$)
(2) incomparability ($\rbrule \not\preorder \rbrule'$ and $\rbrule' \not\preorder \rbrule$),
and
(3) equal rank ($\rbrule \preorder \rbrule'$ and $\rbrule' \preorder \rbrule$).
This structure generalizes standard weighted formulations,
which correspond to placing all rules at equal rank, and
classical lexicographic formulations, which correspond to a total priority ordering.
As shown in \cite{Wongpiromsarn:2026:Formal}, the classical rulebook formalism
generalizes temporal-logic-based planning and optimization-based control in the deterministic
setting.

It can be proved \cite{Censi:2019:Liability} that a rulebook $\rulebook$ induces a preorder
$\rbleq$ on $\realizations$.
The interpretation of $x \rbleq y$ is that $x$ is at least as good as $y$.
Formally, $x \rbleq y$ if, for every rule $\rbrule \in \ruleset$ for which
$\rbrule(x) > \rbrule(y)$, there exists a higher-priority rule $\rbrule' > \rbrule$ such that
$\rbrule'(x) < \rbrule'(y)$.
We say that $x$ is ``strictly better than'' $y$, denoted by $x \rblt y$,
if $x \rbleq y$ but $y \not\rbleq x$.


\section{Problem Formulation}
Consider a system operating under uncertainty.
Let $(\outcomes,\algebra,\probability)$ be a probability space that models the uncertainty, which
may arise from system dynamics, sensing, state-estimation and prediction errors, the behavior of
external agents, or other factors.
Let $\policies \subset \reals^{d}$ be a set of policies available to the system.

Consider a rulebook $\rulebook = \langle \ruleset, \preorder \rangle$,
where each rule $\rbrule_{i} \in \ruleset$ is associated with a risk measure
$\risk_{i}$ and threshold $\threshold_{i}$.
For each policy $\policy \in \policies$, let
$\rbrule_{i}^{\policy} : \outcomes \to \nnreals$ be a measurable nonnegative random variable
representing the violation of $\rbrule_{i}$ when $\policy$ is executed.
Its distribution captures the uncertainty in
evaluating $\rbrule_{i}$ under policy $\policy$ and is induced by the probability measure
$\probability$ on $(\outcomes,\algebra,\probability)$.
Define the \emph{risk-evaluation function}
$q_{i} : \policies \to \nnreals$ associated with $\rbrule_{i}$ by
$q_{i}(\policy) = \risk_{i}(\rbrule_{i}^{\policy})$.
The value $q_{i}(\policy)$ represents the risk incurred by policy $\policy$ with respect to rule
$\rbrule_{i}$ as quantified by the risk measure $\risk_{i}$.
Note that $\risk_{i}$ may differ across rules and can be any risk measure,
including expected value, worst-case, VaR, or CVaR.

The threshold $\threshold_{i}$ specifies the largest acceptable risk for rule $\rbrule_{i}$.
Risks at or below this threshold are treated as equally acceptable for $\rbrule_{i}$,
allowing lower-priority rules to influence the decision once an acceptable level of performance
with respect to $\rbrule_{i}$ has been achieved.
We say that $\policy \in \policies$ \emph{satisfies} rule $\rbrule_{i}$ if
$q_{i}(\policy) \leq \threshold_{i}$.
A rule $\rbrule_{i}$ is \emph{satisfiable on a set} $\policies' \subseteq \policies$ if some
$\policy \in \policies'$ satisfies it.
For convenience, define the \emph{excess risk function}
$\rrrule_{i} : \policies \to \nnreals$ by
$\rrrule_{i}(\policy) = [q_{i}(\policy) - \threshold_{i}]_+$.
This value is zero exactly when $\policy$ satisfies $\rbrule_{i}$, and otherwise equals the amount
by which its risk exceeds $\threshold_{i}$.
Policies are compared according to the rulebook preorder.

\begin{definition}
  A policy $\policy \in \policies$ is \emph{at least as good as}
  $\policy' \in \policies$, denoted by
  $\policy \rbleq \policy'$,
  if for every rule $\rbrule_{i} \in \ruleset$ such that
  $\rrrule_{i}(\policy) > \rrrule_{i}(\policy')$,
  there exists a higher-priority rule $\rbrule_{j} > \rbrule_{i}$ such that
  $\rrrule_{j}(\policy) < \rrrule_{j}(\policy')$.
\end{definition}

\begin{definition}
  A policy $\policy \in \policies$ is \emph{strictly better than}
  $\policy' \in \policies$, denoted by
  $\policy \rblt \policy'$,
  if $\policy \rbleq \policy'$ but
  $\policy' \not\rbleq \policy$.
\end{definition}

\begin{problem}[Risk-aware optimal control with rulebooks]
  \label{prob:risk-aware}
  Compute an optimal policy $\policy^{\star} \in \policies$ such that
  $\policy \not\rblt \policy^{\star}$ for every $\policy \in \policies$.
\end{problem}

We make the following assumptions to ensure that Problem~\ref{prob:risk-aware}
is well posed (see Proposition~\ref{prop:well-posedness}).

\begin{assumption}
  \label{ass:compact}
  The policy space $\policies \subset \reals^{d}$ is compact and nonempty.
\end{assumption}

\begin{assumption}
  \label{ass:lipschitz}
  For every rule $\rbrule_{i} \in \ruleset$, the risk-evaluation function
  $q_{i} : \policies \to \nnreals$ is Lipschitz continuous with known constant $L_{i}$, i.e.,
  $|q_{i}(\policy) - q_{i}(\policy')| \leq L_{i} \|\policy - \policy'\|_{2}$
  for all $\policy, \policy' \in \policies$.
\end{assumption}


\begin{remark}
  As proved in \cite{Wongpiromsarn:2026:Risk},
  the relation $\rbleq$ on
  $\policies$ is a preorder.
  Hence the induced strict relation $\rblt$ is acyclic, which prevents preference cycles
  (e.g., policy $\policy_{1}$ is strictly better than $\policy_{2}$, $\policy_{2}$ is strictly
  better than $\policy_{3}$ but $\policy_{3}$ is strictly better than $\policy_{1}$)
  and gives a well-defined notion of optimality.
\end{remark}

\section{Anytime Lexicographic Optimization via Filtering and Branch-and-Bound}
As shown in \cite{Wongpiromsarn:2026:Formal}, any rulebook $\rulebook$ can be refined into a rulebook
$\rulebook_T=\langle \ruleset_T,\preorder_T\rangle$,
whose priority relation $\preorder_T$ is a total order.
The refined rulebook is consistent with the original in the sense
that no preference from $\rulebook$ is contradicted, i.e.,
for any two policies $\policy, \policy' \in \policies$,
$\policy \rbleq \policy' \Longrightarrow \policy \preorder_{\rulebook_{T}} \policy'$ and
$\policy \rblt \policy' \Longrightarrow \policy <_{\rulebook_{T}} \policy'$.
Although policies that are incomparable under $\rulebook$ may become comparable
under $\rulebook_{T}$, an optimal policy for $\rulebook_T$ is also optimal for $\rulebook$.

In the remainder of the paper, we assume without loss of generality that the rulebook
has been refined into a total order with
$\ruleset_{T} = \{\rbrule_1,\ldots,\rbrule_m\}$ and
$\rbrule_{i} >_{T} \rbrule_{i+1}$ for all $i$.
Problem~\ref{prob:risk-aware} can then be formulated as a
lexicographic optimization problem over the excess-risk vector
$\rrrule(\policy)=
\bigl(\rrrule_1(\policy),\ldots,\rrrule_m(\policy)\bigr)$.
Since each $q_i$ is Lipschitz continuous, the corresponding excess risk function
$\rrrule_i$ is continuous.

Given any rule level $i \in \{1, \ldots, m\}$ and compact subset $\policies' \subseteq \policies$,
let $\alpha_i(\policies')=\min_{\policy\in \policies'}\rrrule_i(\policy)$
denote the smallest excess risk with respect to rule $\rbrule_{i}$ achievable over the policy set
$\policies'$.
Since $\policies'$ is compact and $\rrrule_{i}$ is continuous,
the extreme value theorem ensures that the minimum is
attained and hence $\alpha_i(\policies')$ is well-defined \cite{Royden:2023:Real}.
The lexicographic structure naturally induces a sequence of nested sets of
candidate policies.
Let $\policies^{\star}_{0} = \policies$, and for each $i \in \{1, \ldots, m\}$, define
$\policies^{\star}_{i} = \{\policy \in \policies^{\star}_{i-1} \mid
\rrrule_{i}(\policy) = \alpha_{i}(\policies^{\star}_{i-1})\}$.
In other words, $\policies^{\star}_{i}$ is the set of policies that are lexicographically optimal
with respect to the first $i$ rules and
the rulebook-optimal policy set is $\policies^{\star}_m$.
Starting from the nonempty compact set $\policies$,
the extreme value theorem and continuity of $\rrrule_i$
ensure that each $\policies^{\star}_{i}$ is compact and nonempty,
so the construction is well defined.
Together with the consistency of the refinement,
this establishes existence of an optimal policy for Problem~\ref{prob:risk-aware}.

\begin{proposition}
  \label{prop:well-posedness}
  Problem~\ref{prob:risk-aware} admits an optimal policy.
\end{proposition}

In many applications, high-priority rules correspond to safety requirements
and are associated with thresholds that define acceptable levels of risk.
Once the risk falls below the prescribed threshold,
further reducing it provides no additional preference under that rule.
Lower-priority rules then determine the preference among policies
satisfying the higher-priority rule.
Optimization is required only for rules whose thresholds cannot be achieved.
Formally, let
$\policies_i=\{\policy\in\policies \mid q_i(\policy)\leq \threshold_i\}$
denote the set of policies that satisfy rule $\rbrule_i$.
If $\rbrule_i$ is satisfiable on $\policies^{\star}_{i-1}$,
then $\alpha_i(\policies^{\star}_{i-1})=0$ and
$\policies^{\star}_i=\policies^{\star}_{i-1}\cap \policies_i$.
Thus, all policies in $\policies^{\star}_{i-1}\cap \policies_i$
are considered equally acceptable with respect to $\rbrule_{i}$
and the $i$th step is reduced to a feasibility-filtering step.
On the other hand, if $\rbrule_i$ is not satisfiable on $\policies^{\star}_{i-1}$, then
$\rrrule_i(\policy)=q_i(\policy)-\threshold_i$ on $\policies^{\star}_{i-1}$, so minimizing
$\rrrule_i$ is equivalent to minimizing $q_i$.

This observation motivates an algorithm that combines viability-style feasibility filtering with
Lipschitz branch-and-bound optimization to progressively restrict the set of candidate policies as
it processes the rule hierarchy.
The main design principle is to keep two kinds of information:
a box collection, which provides a certified lower bound,
and certified ``incumbents'', i.e., concrete feasible policies already found by the algorithm
(using branch-and-bound terminology),
which can be returned by the algorithm and provide a certified upper bound.

\subsection{Admissible Domains, Box Bounds and Certificates}
Under a finite computational budget, computing the exact sets
$\policies^{\star}_i$ is generally infeasible,
since it requires solving a global optimization or global level-set characterization problem for a
black-box risk-evaluation function.
We therefore compute a sequence of certified admissible domains that
contain the exact lexicographic solution sets.

\begin{definition}
  The \emph{certified admissible domains} are defined recursively by
  $D_{0}=\policies$ and
  \begin{equation}
    D_i
    =
    \{\policy \in D_{i-1} \mid \rrrule_i(\policy) \le \overline{\alpha}_i\},
    \label{eq:certified-domain}
  \end{equation}
  where $\overline\alpha_i \in \nnreals$ satisfies $\overline\alpha_i\ge\alpha_i(D_{i-1})$.
\end{definition}

By construction, each certified admissible domain contains the corresponding exact
lexicographic solution set, i.e.,
$\policies^{\star}_i\subseteq D_i$.
To quantify this relaxation, define
$\varepsilon_i=\overline\alpha_i-\alpha_i(D_{i-1})$.
Then, $D_{i}$ is precisely the
$\varepsilon_i$-optimal set for rule $\rbrule_{i}$ over the certified admissible domain $D_{i-1}$.
Because $\alpha_i(D_{i-1})$ is generally unknown,
the algorithm computes a certified upper bound on $\varepsilon_i$.
If $\varepsilon_i=0$ for every rule level,
then each certified admissible domain coincides with the exact lexicographic solution set, i.e.,
$D_{i} = \policies^{\star}_{i}$ for all rule levels $i$..

To compute $D_{i}$,
the algorithm represents subsets of $\reals^{d}$ by finite collections of closed axis-aligned boxes.
For a box $B\subseteq\reals^d$, let $c(B)$ be its center and define its covering radius as
$\radius(B)=\frac{\sqrt d}{2}l(B)$, where
$l(B) = \max_{j} l_{j}(B)$ is the length of the longest side of $B$ and $l_j(B)$ is the length
of its $j$th side.
Every point in $B$ lies within Euclidean distance $\radius(B)$ from $c(B)$.

For every rule $\rbrule_i$ and box $B$, define the box bounds
\begin{equation}
  \begin{split}
    \underline q_i(B) &= q_i(c(B))-L_i\radius(B),\\
    \overline q_i(B) &= q_i(c(B))+L_i\radius(B),\\
    \underline{\rrrule}_i(B) &= [\underline q_i(B)-\threshold_i]_{+}, \qquad
                               \overline{\rrrule}_i(B) = [\overline q_i(B)-\threshold_i]_{+},\\
    w_i(B) &= \overline q_i(B)-\underline q_i(B)=2L_i\radius(B)=L_i\sqrt d\,l(B).
  \end{split}
  \label{eq:LBUB}
\end{equation}
This construction assumes that $q_i(\policy)$ can be queried at policy points $c(B)$ for each
box $B$.
As shown in the following lemma,
$\underline q_i(B)$ and $\overline q_i(B)$ are certified lower and upper bounds on the
risk-evaluation function $q_i$ over $B\cap\policies$,
while $\underline{\rrrule}_i(B)$ and $\overline{\rrrule}_i(B)$ are
the corresponding bounds on $\rrrule_{i}$ and
certify how much a policy in $B\cap\policies$ can violate $\rbrule_i$ beyond the acceptable level
$\threshold_i$.
Finally, $w_i(B)$ quantifies the uncertainty width induced by Lipschitz continuity of $q_i$ over
$B$.

\begin{lemma}
  \label{lem:lipschitz-bounds}
  Suppose $c(B)\in\policies$ and
  $\underline q_i(B)$, $\overline q_i(B)$, $\underline{\rrrule}_i(B)$, $\overline{\rrrule}_i(B)$
  and $w_i(B)$ are defined by~\eqref{eq:LBUB}.
  Then, for every policy $\policy\in B\cap\policies$,
  \begin{eqnarray}
    \label{eq:q-bounds}
    &\underline q_i(B)
      \leq q_i(\policy)
      \leq \overline q_i(B),\\
    \label{eq:v-bounds}
    &\underline{\rrrule}_i(B)
      \leq \rrrule_i(\policy)
      \leq \overline{\rrrule}_i(B),\\
    \label{eq:v-width}
    &\overline{\rrrule}_i(B)-\underline{\rrrule}_i(B)\leq w_i(B).
  \end{eqnarray}
\end{lemma}
\begin{proof}
  Since $c(B)\in\policies$ and $\policy\in\policies$,
  Assumption~\ref{ass:lipschitz} gives
  $|q_i(\policy)-q_i(c(B))|\leq L_i\|\policy-c(B)\|_2$.
  Since $\policy\in B$, $\|\policy-c(B)\|_2\leq\radius(B)$.
  Rearranging yields \eqref{eq:q-bounds}.
  Equation~\eqref{eq:v-bounds} follows from \eqref{eq:q-bounds} because
  $x\mapsto\max\{x-\threshold_i,0\}$ is monotone.
  Finally, by \eqref{eq:LBUB},
  $\overline q_i(B)-\underline q_i(B)=w_i(B)$.
  Since
  $[b]_{+}-[a]_{+}\le b-a$ for all $a\le b$,
  we obtain
  $\overline{\rrrule}_i(B)-\underline{\rrrule}_i(B)
  \le
  \overline q_i(B)-\underline q_i(B)
  =
  w_i(B)$.
\end{proof}

The algorithm does not depend on the particular construction in~\eqref{eq:LBUB},
which uses the box center $c(B)$ to compute the bounds.
Any construction is valid if the bounds satisfy
\eqref{eq:q-bounds}--\eqref{eq:v-width} for every policy $\policy\in B\cap\policies$ as
all correctness results rely only on these inequalities,
not how the bounds are computed.

The algorithm processes the rules sequentially.
As will be explained in Section~\ref{sec:alg},
when processing rule level $i$, it starts with a box collection $\mathcal{O}$
satisfying $D_{i-1}\subseteq \bigcup_{B\in\mathcal O}B$
together with a nonempty finite certified incumbent set $\mathcal C\subseteq D_{i-1}$.
It then iteratively prunes and refines $\mathcal{O}$.
Pruning removes boxes in $\mathcal O$ that cannot contain a minimizer of $\rrrule_{i}$ over
$D_{i-1}$, whereas refinement improves the resolution of $\mathcal O$.
Throughout this process, the algorithm ensures that $\mathcal O$ always covers all minimizers of
$\rrrule_i$ over $D_{i-1}$, i.e.,
$\argmin_{\policy\in D_{i-1}}\rrrule_i(\policy)\subseteq \bigcup_{B\in\mathcal O} B$.

To support this pruning and refinement process,
the algorithm maintains a lower certificate $\underline\alpha_i$,
an upper certificate $\overline\alpha_i$, and a certified optimality gap $\varepsilon_i$.
At the beginning of rule level $i$, these certificates are initialized as
\begin{equation}
  \underline\alpha_i=\min_{B\in\mathcal O}\underline{\rrrule}_i(B), \
  \overline\alpha_i = \min_{\policy\in\mathcal C}\rrrule_i(\policy), \
  \varepsilon_i = \overline\alpha_i-\underline\alpha_i.
  \label{eq:anytime-gap}
\end{equation}
As the box collection $\mathcal O$ is pruned and refined,
the current box-based lower bound $\min_{B\in\mathcal O}\underline{\rrrule}_i(B)$
may increase or decrease.
To ensure that $\underline\alpha_{i}$ is monotone nondecreasing,
the algorithm maintains $\underline\alpha_{i}$ as a running lower certificate
by updating it only when the newly computed box-based lower bound is larger than its current value:
\begin{equation}
  \underline\alpha_i \leftarrow
  \max\{\underline\alpha_i, \min_{B\in\mathcal O}\underline{\rrrule}_i(B)\}.
  \label{eq:cert-update}
\end{equation}
Since every box collection $\mathcal O$ encountered during
the pruning and refinement process continues to cover all minimizers
of $\rrrule_i$ over $D_{i-1}$,
each newly computed value of $\min_{B\in\mathcal O}\underline{\rrrule}_i(B)$
is a valid lower bound on $\alpha_{i}(D_{i-1})$
and so is their maximum, $\underline\alpha_i$.
In contrast, the certified incumbent set $\mathcal C$ only grows over time,
so simply recomputing $\overline \alpha_i$ and $\varepsilon_i$ using \eqref{eq:anytime-gap}
at each update is sufficient to ensure that $\overline\alpha_i$
and $\varepsilon_i$ are monotone nonincreasing,
without requiring the analogous running update used for $\underline\alpha_i$.

\begin{lemma}
  \label{lem:gap-certificate}
  Suppose that every initialization and update of
  $\underline\alpha_i$ and $\overline\alpha_i$ through
  \eqref{eq:anytime-gap} and \eqref{eq:cert-update}
  uses sets satisfying
  $\emptyset\neq\mathcal C\subseteq D_{i-1}$ and
  $\argmin_{\policy\in D_{i-1}}\rrrule_i(\policy) \subseteq \bigcup_{B\in\mathcal O} B$.
  Then,
  \[
    \underline\alpha_i \leq \alpha_i(D_{i-1}) \leq \overline\alpha_i.
  \]
\end{lemma}
\begin{proof}
  Consider an arbitrary $\policy^{\star} \in \argmin_{\policy\in D_{i-1}}\rrrule_i(\policy)$, i.e.,
  $\rrrule_{i}(\policy^{\star}) = \alpha_{i}(D_{i-1})$.
  For each certificate update $k$ during rule level $i$, let $\mathcal O(k)$ and $\mathcal C(k)$
  denote the corresponding current sets, and define
  $\beta(k) = \min_{B\in\mathcal O(k)}\underline{\rrrule}_i(B)$.
  By assumption,
  $\argmin_{\policy\in D_{i-1}}\rrrule_i(\policy)\subseteq \bigcup_{B\in\mathcal O(k)} B$,
  so there exists $B^{\star}\in\mathcal O(k)$ with $\policy^\star\in B^{\star}$.
  Hence
  \[
    \beta(k)
    = \min_{B\in\mathcal O(k)}\underline{\rrrule}_i(B)
    \leq \underline{\rrrule}_i(B^{\star})
    \leq \rrrule_i(\policy^\star)
    =\alpha_i(D_{i-1}).
  \]
  Since $\underline\alpha_i = \max_{k} \beta(k)$,
  with $\beta(k) \leq \alpha_i(D_{i-1})$ for all $k$,
  it follows that
  $\underline\alpha_i\leq \alpha_i(D_{i-1})$.
  Finally, since $\mathcal C\subseteq D_{i-1}$,
  it follows that
  $\overline\alpha_i
  =
  \min_{\policy\in\mathcal C}\rrrule_i(\policy)
  \ge
  \min_{\policy\in D_{i-1}}\rrrule_i(\policy)
  =
  \alpha_i(D_{i-1})$,
  which completes the proof.
\end{proof}


Based on \eqref{eq:v-bounds} and Lemma \ref{lem:gap-certificate},
$\underline{\rrrule}_{i}(B)$ provides a certified lower bound on
$\rrrule_{i}(\policy)$ for any $\policy \in B\cap\policies$,
while $\overline\alpha_i$ is a certified upper bound on the optimal
value of $\rrrule_{i}(\policy)$ for any $\policy \in D_{i-1}$.
As a result, if $\underline{\rrrule}_i(B) > \overline\alpha_{i}$,
then $B$ cannot contain a minimizer of $\rrrule_{i}$ over $D_{i-1}$.
This motivates the following notion of a competitive box and the associated sound pruning result.

\begin{definition}
  A box $B\in\mathcal O$ is \emph{competitive} for rule $\rbrule_i$ if
  $\underline{\rrrule}_i(B)\leq\overline\alpha_i$.
\end{definition}

\begin{lemma}
  \label{lem:sound-pruning}
  Suppose $\emptyset\neq\mathcal C\subseteq D_{i-1}$.
  Then, any box $B$ with $\underline{\rrrule}_i(B)>\overline\alpha_i$ satisfies
  $B\cap\argmin_{\policy\in D_{i-1}}\rrrule_i(\policy) = \emptyset$.
\end{lemma}
\begin{proof}
  For every $\policy\in B\cap\policies$, \eqref{eq:v-bounds} gives
  $\rrrule_i(\policy)\geq\underline{\rrrule}_i(B)>\overline\alpha_i$.
  By definition of $\overline\alpha_i$ in \eqref{eq:anytime-gap}
  and the assumption that $\emptyset\neq\mathcal C \subseteq D_{i-1}$, there exists an incumbent
  $\policy^{\rm inc}\in D_{i-1}$ with $\rrrule_i(\policy^{\rm inc})=\overline\alpha_i$.
  Therefore, for every $\policy\in B\cap\policies$, $\rrrule_i(\policy) > \rrrule_i(\policy^{\rm inc})
  \geq \min_{\policy\in D_{i-1}}\rrrule_i(\policy)$.
  Therefore,
  $B\cap\argmin_{\policy\in D_{i-1}}\rrrule_i(\policy) =\emptyset$.
\end{proof}

For a satisfiable rule, $\alpha_i(D_{i-1})=0$.
Boxes with $\overline q_i(B)\leq\threshold_i$ are certified feasible, boxes with
$\underline q_i(B)>\threshold_i$ are certified infeasible, and the remaining boxes are ambiguous.
Once $\overline\alpha_i=0$ has been certified, the pruning test removes the certified infeasible boxes
and leaves only certified feasible or ambiguous boxes.
For an unsatisfiable rule, the same pruning test is branch-and-bound optimization on
$\rrrule_i$.
The following lemma provides the stopping criterion for a satisfiable rule.
Once every box in $\mathcal O$ satisfies $\overline q_i(B)\leq\threshold_i$,
the set $D_{i-1}\cap\policies_i$ has been characterized exactly,
so no further subdivision is required at rule level $i$.

\begin{lemma}
  \label{lem:threshold-stopping}
  Suppose $D_{i-1}\subseteq\bigcup_{B\in\mathcal O}B$ and let
  $\mathcal O^+=\{B\in\mathcal O\mid \underline q_i(B)\leq\threshold_i\}$.
  Then, $D_{i-1}\cap\policies_i \subseteq D_{i-1}\cap\bigcup_{B\in\mathcal O^+}B$.
  If, in addition, $\overline q_i(B)\leq\threshold_i$ for every $B\in\mathcal O^+$,
  then
  $D_{i-1}\cap\policies_{i} =D_{i-1}\cap\bigcup_{B\in\mathcal O^+}B$.
\end{lemma}
\begin{proof}
  We first show that
  $D_{i-1}\cap\policies_i \subseteq D_{i-1}\cap\bigcup_{B\in\mathcal O^+}B$.
  Consider an arbitrary $\policy\in D_{i-1}\cap\policies_i$.
  Since $\mathcal O$ covers $D_{i-1}$, there exists a box
  $B\in\mathcal O$ such that $\policy\in B$.
  By \eqref{eq:q-bounds},
  $\underline q_i(B)\le q_i(\policy)$.
  Since $\policy \in \policies_{i}$, by definition,
  $q_i(\policy)\le\threshold_i$.
  Combining these inequalities, we get
  $\underline q_i(B)\le q_i(\policy)\le\threshold_i$.
  Hence, by definition, $B\in\mathcal O^+$ and therefore
  $\policy\in\bigcup_{B\in\mathcal O^+}B$.
  This proves the first inclusion.

  Now suppose that every box $B\in\mathcal O^+$ satisfies
  $\overline q_i(B)\le\threshold_i$.
  We show the reverse inclusion.
  Consider an arbitrary $\policy\in D_{i-1}\cap\bigcup_{B\in\mathcal O^+}B$.
  Then $\policy\in D_{i-1}$ and $\policy\in B$
  for some $B\in\mathcal O^+$.
  Applying \eqref{eq:q-bounds} again gives
  $q_i(\policy)\le\overline q_i(B)\le\threshold_i$,
  so $\policy\in\policies_i$.
  Therefore,
  $D_{i-1}\cap\bigcup_{B\in\mathcal O^+}B \subseteq D_{i-1}\cap\policies_i$.
  The two inclusions together imply
  $D_{i-1}\cap\policies_{i} =D_{i-1}\cap\bigcup_{B\in\mathcal O^+}B$.
\end{proof}

\subsection{Certified Anytime Algorithm}
\label{sec:alg}
Algorithm~\ref{alg:anytime-rulebook} implements the recursion in \eqref{eq:certified-domain}.
Its state at the beginning of rule level $i$ is a certified admissible domain
$D_{i-1}$, a box collection $\mathcal O$, and a finite certified incumbent set $\mathcal C$.
As proved later in Lemma~\ref{lem:level-invariants}, these sets satisfy
$D_{i-1}\subseteq\bigcup_{B\in\mathcal O}B$ and
$\emptyset \not= \mathcal C\subseteq D_{i-1}$.
Lines~\ref{alg:init-domain}--\ref{alg:init-incumbents} initialize
these sets with $D_0=\policies$.
The initial box collection $\mathcal O$ may be any finite collection of
closed boxes whose union contains $\policies$, e.g., a bounding box of $\policies$.
Such a collection exists because by Assumption~\ref{ass:compact},
the set $\policies$ is compact.
The initial certified incumbent set $\mathcal C$ may be any nonempty finite subset of $\policies$,
e.g., a sampled or user-provided policy in $\policies$.
The nonemptiness of $\mathcal C$ is needed only to define a finite upper certificate
$\overline\alpha_i$.

The helper \textsc{Bounds} in lines
\ref{alg:compute-bounds} and \ref{alg:child-bounds} computes and stores
$\underline q_i(B)$, $\overline q_i(B)$, $\underline{\rrrule}_i(B)$,
$\overline{\rrrule}_i(B)$, and $w_i(B)$ for each box $B$ in the first argument.
It may use any certificate satisfying~\eqref{eq:q-bounds}--\eqref{eq:v-width},
e.g., the center construction~\eqref{eq:LBUB} when applicable.
The helper \textsc{InitCert} in
line~\ref{alg:init-cert} initializes the certificates
$\overline\alpha_i$, $\underline\alpha_i$,
and $\varepsilon_i$ according to~\eqref{eq:anytime-gap}.
The helper \textsc{UpdateCert} in line~\ref{alg:update-cert} updates
$\underline\alpha_i$ according to
\eqref{eq:cert-update}, while recomputing $\overline\alpha_i$ and $\varepsilon_i$ according
to~\eqref{eq:anytime-gap}.

The algorithm employs three routines that prune box collection $\mathcal O$.
When $\overline\alpha_i>0$,
line~\ref{alg:positive-prune} calls competitive pruning, \textsc{CPrune}$(\mathcal O,\overline\alpha_i,i)$,
which returns the set of competitive boxes
$\{B\in\mathcal O\mid \underline{\rrrule}_i(B)\leq\overline\alpha_i\}$, as justified by
Lemma~\ref{lem:sound-pruning}.
When $\overline\alpha_i=0$, then Lemma~\ref{lem:gap-certificate} and nonnegativity of $\rrrule_i$
imply $\alpha_i(D_{i-1})=0$.
In this case, line~\ref{alg:threshold-prune} calls threshold pruning,
\textsc{TPrune}$(\mathcal O,i)$,
which returns
$\{B\in\mathcal O\mid \underline q_i(B)\leq\threshold_i\}$, as justified by
Lemma~\ref{lem:threshold-stopping}.
Finally, domain pruning, \textsc{DPrune}$(\mathcal O,i)$ in
lines~\ref{alg:domain-prune-init} and \ref{alg:domain-prune-update}
removes boxes that are not competitive for some previous rule level $j < i$.
Specifically, for $i=1$,
\textsc{DPrune}$(\mathcal O, i)$ simply returns $\mathcal O$.
For $i>1$, it returns
$\{B\in\mathcal O\mid\underline{\rrrule}_j(B)\leq\overline\alpha_j,\ \forall j<i\}$.
Although \textsc{CPrune} and \textsc{TPrune} remove boxes that are not competitive for the current
rule level $i$, bisection (lines~\ref{alg:select-box}--\ref{alg:update-cover})
may create child boxes that are no longer competitive for an earlier rule
level $j<i$, even when their parent box was.
Thus, \textsc{DPrune} is invoked before \textsc{InitCert} and \textsc{UpdateCert} to remove
such boxes and ensure that  $\underline\alpha_{i}$ is computed only from boxes that have not already
been certified to lie outside $D_{i-1}$.

The set $\mathcal R$ of refinable boxes is constructed in line~\ref{alg:set-refine-positive}
when $\overline\alpha_{i} > 0$ and line~\ref{alg:threshold-refine} when $\overline\alpha_{i} = 0$.
When $\overline\alpha_{i} > 0$, $\mathcal R$ is simply the box collection
remaining after \textsc{CPrune}.
When $\overline\alpha_{i} = 0$, line~\ref{alg:threshold-refine} calls
\textsc{TRefine}$(\mathcal O,i)$,
which returns the set of ambiguous boxes
$\{B\in\mathcal O\mid \overline q_i(B)>\threshold_i\}$.
If no ambiguous box remains, line~\ref{alg:threshold-stop} terminates the current level using
Lemma~\ref{lem:threshold-stopping}.
Otherwise, lines~\ref{alg:select-box}--\ref{alg:update-cover}
bisect the widest refinable box along the longest dimension.
Specifically, suppose $B^{\star}=\prod_{k=1}^{d}[a_k,b_k]$.
Choose the longest side $j\in\argmax_k(b_k-a_k)$.
Split at the middle $m_j=(a_j+b_j)/2$ and
generate
$B^- = \bigl(\prod_{k<j}[a_k,b_k]\bigr)\times[a_j,m_j]\times
\bigl(\prod_{k>j}[a_k,b_k]\bigr)$
and
$B^+ = \bigl(\prod_{k<j}[a_k,b_k]\bigr)\times[m_j,b_j]\times
\bigl(\prod_{k>j}[a_k,b_k]\bigr)$.

The helper \textsc{CSamples}$(\{B^{+}, B^{-}\}, D_{i-1})$
in line~\ref{alg:sample-incumbents} returns a finite set of sampled
policies from $B^{+}$ and $B^{-}$ that are certified to lie in $D_{i-1}$.
The sampling procedure may select any policies from $B^{+}$ and $B^{-}$, e.g., their centers, and
returns only those certified to lie in $D_{i-1}$.
For $i=1$, certification simply requires checking that a sampled policy $\policy$
lies in $\policies$.
For $i>1$, this can be done by checking that $\rrrule_j(\policy)\leq\overline\alpha_j$ for all $j<i$.
For efficiency, \textsc{CSamples} may certify an entire box by checking that
$\overline{\rrrule}_j(B)\leq\overline\alpha_j$ for all $j<i$,
which guarantees that every policy in $B$ lies in $D_{i-1}$.

After processing rule level $i$,
line~\ref{alg:set-domain} calls \textsc{Propagate}, which conceptually returns the set in
\eqref{eq:certified-domain} and line~\ref{alg:filter-incumbents} filters the incumbents to
$\mathcal C\cap D_i$ before the next level.
In practice, however, the algorithm does not have to construct $D_{i}$ explicitly.
Since $D_{i}$ is only used for membership tests
(in \textsc{CSamples} and line~\ref{alg:filter-incumbents}),
it suffices to store only the threshold $\overline\alpha_i$,
from which membership in $D_{i}$ can be determined.
Also, note that the filtered set $\mathcal C\cap D_i$
in line~\ref{alg:filter-incumbents} is nonempty because
$\argmin_{\policy\in\mathcal C}\rrrule_i(\policy)\subseteq \mathcal C\cap D_i$.
Finally, \textsc{Best}$(\mathcal C,m)$
in line~\ref{alg:return} returns any element of $\argmin_{\policy\in\mathcal C}\rrrule_m(\policy)$.

\begin{algorithm}[t]
  \caption{Anytime Optimization with Rulebooks}
  \label{alg:anytime-rulebook}
  \begin{algorithmic}[1]
    \State $D_0\gets\policies$ \label{alg:init-domain}
    \State initialize finite box collection $\mathcal O$ such that $D_{0} \subseteq \bigcup_{B\in\mathcal O}B$ \label{alg:init-cover}
    \State initialize nonempty finite $\mathcal C\subseteq D_0$ \label{alg:init-incumbents}
    \For{$i=1,\ldots,m$} \label{alg:for-level}
      \State $\mathcal O\gets\Call{DPrune}{\mathcal O,i}$ \label{alg:domain-prune-init}
      \State \Call{Bounds}{$\mathcal O,i$} \label{alg:compute-bounds}
      \State \Call{InitCert}{$\mathcal O,\mathcal C,i$} \label{alg:init-cert}
      \While{computational budget remains} \label{alg:while}
        \If{$\overline\alpha_i>0$} \label{alg:if-not-zero}
          \State $\mathcal O\gets\Call{CPrune}{\mathcal O,\overline\alpha_i,i}$ \label{alg:positive-prune}
          \State $\mathcal R\gets\mathcal O$ \label{alg:set-refine-positive}
        \Else
          \State $\mathcal O\gets\Call{TPrune}{\mathcal O,i}$ \label{alg:threshold-prune}
          \State $\mathcal R\gets\Call{TRefine}{\mathcal O,i}$ \label{alg:threshold-refine}
          \If{$\mathcal R=\emptyset$} \label{alg:if-no-refine}
            \State \textbf{break} \label{alg:threshold-stop}
          \EndIf
        \EndIf \label{alg:refinement-stop}
        \State $B^\star\gets\argmax_{B\in\mathcal R} w_i(B)$ \label{alg:select-box}
        \State $(B^+,B^-)\gets\Call{Bisect}{B^\star}$ \label{alg:split-box}
        \State $\mathcal O\gets\Call{Replace}{\mathcal O,B^\star,\{B^+,B^-\}}$ \label{alg:update-cover}
        \State $\mathcal O\gets\Call{DPrune}{\mathcal O,i}$ \label{alg:domain-prune-update}
        \State $\mathcal S\gets\Call{CSamples}{\{B^+, B^-\},D_{i-1}}$ \label{alg:sample-incumbents}
        \State $\mathcal C\gets\mathcal C\cup\mathcal S$ \label{alg:add-incumbents}
        \State \Call{Bounds}{$\{B^+,B^-\},i$} \label{alg:child-bounds}
        \State \Call{UpdateCert}{$\mathcal O,\mathcal C,i$} \label{alg:update-cert}
      \EndWhile
      \State $D_i\gets\Call{Propagate}{D_{i-1},\overline\alpha_i,i}$ \label{alg:set-domain}
      \State $\mathcal C\gets\mathcal C\cap D_i$ \label{alg:filter-incumbents}
    \EndFor
    \State return \Call{Best}{$\mathcal C,m$} and gaps $\varepsilon_{1:m}$ \label{alg:return}
  \end{algorithmic}
\end{algorithm}

\subsection{Correctness of Algorithm~\ref{alg:anytime-rulebook}}
To prove the correctness of Algorithm~\ref{alg:anytime-rulebook},
we begin by showing that the pruning step preserves the coverage property required by the subsequent
correctness proof.
For $a\in\mathbb R$, define
$\widehat D_i(a)=\{\policy\in D_{i-1}\mid \rrrule_i(\policy)\leq a\}$.

\begin{lemma}
  \label{lem:current-pruning-preserves-cover}
  Suppose $\widehat D_i(\overline\alpha_i)\subseteq\bigcup_{B\in\mathcal O}B$ holds
  immediately before the conditional statement (line~\ref{alg:if-not-zero}).
  Then, the same inclusion holds immediately after
  the execution of \textsc{CPrune} (line~\ref{alg:positive-prune}) when $\overline\alpha_{i} > 0$,
  and immediately after the execution of \textsc{TPrune} (line~\ref{alg:threshold-prune})
  when $\overline\alpha_{i} = 0$.
\end{lemma}
\begin{proof}
  First, suppose $\overline\alpha_{i}>0$.
  Consider an arbitrary $\policy\in\widehat D_i(\overline\alpha_i)$.
  By definition,
  $\rrrule_i(\policy)\leq \overline\alpha_{i}$.
  Also, by assumption, $\widehat D_i(\overline\alpha_i)\subseteq\bigcup_{B\in\mathcal O}B$,
  so there exists $B\in\mathcal O$ with $\policy\in B$
  immediately before line~\ref{alg:positive-prune}.
  By~\eqref{eq:v-bounds},
  $\underline{\rrrule}_i(B)\leq\rrrule_i(\policy)\leq \overline\alpha_{i}$,
  so \textsc{CPrune} retains $B$.
  Hence, $\widehat D_i(\overline\alpha_i)\subseteq\bigcup_{B\in\mathcal O}B$ after
  line~\ref{alg:positive-prune}.

  Now suppose $\overline\alpha_{i}=0$.
  Then, $\widehat D_i(\overline\alpha_{i})=D_{i-1}\cap\policies_i$.
  Consider an arbitrary $\policy\in\widehat D_i(\overline\alpha_i)$.
  Since $\policy\in\policies_i$, by definition,
  $q_i(\policy)\leq\threshold_i$.
  Also, by assumption, $\widehat D_i(\overline\alpha_i)\subseteq\bigcup_{B\in\mathcal O}B$,
  so there exists $B\in\mathcal O$ with $\policy\in B$
  immediately before line~\ref{alg:threshold-prune}.
  By~\eqref{eq:q-bounds},
  $\underline q_i(B)\leq q_i(\policy)\leq\threshold_i$,
  so \textsc{TPrune} retains $B$.
  Hence, $\widehat D_i(\overline\alpha_{i})\subseteq\bigcup_{B\in\mathcal O}B$ after
  line~\ref{alg:threshold-prune}.
\end{proof}

We now establish the invariants of a fixed rule level in the next four lemmas.
For these lemmas, fix an arbitrary rule level $i \in \{1, \ldots, m\}$ and
assume that Assumption~\ref{ass:lipschitz} holds, and that,
at the beginning of level $i$, immediately before line~\ref{alg:domain-prune-init},
$D_{i-1}$ is compact,
$D_{i-1}\subseteq\bigcup_{B\in\mathcal O}B$, and
$\mathcal C$ is a nonempty finite subset of $D_{i-1}$.
Let
\begin{equation}
  M_i=\argmin_{\policy\in D_{i-1}}\rrrule_i(\policy),
  \label{eq:Mi}
\end{equation}
denote the set of minimizers of $\rrrule_i$ over $D_{i-1}$.
Since $D_{i-1}$ is compact and $\rrrule_i$ is continuous,
$M_i$ is nonempty and compact.

\begin{lemma}
  \label{lem:Mi}
  Immediately after the execution of \textsc{InitCert} (line~\ref{alg:init-cert})
  and immediately before and after every execution of \textsc{UpdateCert}
  (line~\ref{alg:update-cert}) during level $i$,
  $M_i\subseteq\widehat D_i(\overline\alpha_i)$.
\end{lemma}
\begin{proof}
  Immediately after line~\ref{alg:init-cert}, $\overline\alpha_i$ is computed from $\mathcal
  C\subseteq D_{i-1}$, which is nonempty by assumption.
  Hence $\alpha_i(D_{i-1})\leq\overline\alpha_i$, and therefore
  $M_i\subseteq\widehat D_i(\overline\alpha_i)$.

  During each while-loop iteration, no line before line~\ref{alg:update-cert} changes
  $\overline\alpha_i$.
  Therefore the inclusion that holds immediately after the preceding execution of
  \textsc{InitCert} or \textsc{UpdateCert} still holds immediately before line~\ref{alg:update-cert}.

  Immediately after line~\ref{alg:update-cert}, $\overline\alpha_i$ is recomputed from
  $\mathcal C\subseteq D_{i-1}$, which is nonempty because
  it is nonempty at the beginning of the level by assumption,
  and every preceding execution of line~\ref{alg:add-incumbents} adds
  $\mathcal S\subseteq D_{i-1}$ returned by \textsc{CSamples}.
  Thus $\alpha_i(D_{i-1})\leq\overline\alpha_i$, and hence
  $M_i\subseteq\widehat D_i(\overline\alpha_i)$ immediately after line~\ref{alg:update-cert}.
\end{proof}

\begin{lemma}
  \label{lem:cover-invariant}
  At the beginning of each iteration of the while loop,
  i.e., immediately before line~\ref{alg:if-not-zero},
  \begin{equation}
    \widehat D_i(\overline\alpha_i)
    \subseteq\bigcup_{B\in\mathcal O}B .
    \label{eq:loop-cover-induction}
  \end{equation}
\end{lemma}
\begin{proof}
  We use induction over the iterations of the while loop to show that
  \eqref{eq:loop-cover-induction} holds at the beginning of every iteration.

  Before line~\ref{alg:domain-prune-init}, $D_{i-1}\subseteq\bigcup_{B\in\mathcal O}B$ by
  assumption.
  If $i=1$, \textsc{DPrune} does not remove any box from $\mathcal O$.
  Suppose $i>1$ and consider any $\policy\in D_{i-1}$.
  Since $D_{i-1}\subseteq\bigcup_{B\in\mathcal O}B$, there exists a box $B\in\mathcal O$ with
  $\policy\in B$.
  Since $\policy\in D_{i-1}$,
  $\rrrule_j(\policy)\leq\overline\alpha_j$ for all $j<i$.
  By~\eqref{eq:v-bounds},
  $\underline{\rrrule}_j(B)\leq\rrrule_j(\policy)\leq\overline\alpha_j$ for all $j<i$.
  Thus line~\ref{alg:domain-prune-init} retains $B$.
  Therefore $D_{i-1}\subseteq\bigcup_{B\in\mathcal O}B$ still holds after
  line~\ref{alg:domain-prune-init}.
  Since $\widehat D_i(\overline\alpha_i)\subseteq D_{i-1}\subseteq\bigcup_{B\in\mathcal O}B$,
  \eqref{eq:loop-cover-induction} holds before the first iteration of the while loop.

  The only lines in the while loop that can remove points from
  $\bigcup_{B\in\mathcal O}B$ are the pruning lines
  \ref{alg:positive-prune}, \ref{alg:threshold-prune}, and \ref{alg:domain-prune-update}.
  Note that while line~\ref{alg:update-cover} replaces $B^\star$ by $B^+$ and $B^-$,
  by construction $B^\star=B^+\cup B^-$, so $\bigcup_{B\in\mathcal O}B$ is unchanged.
  Additionally, the only line in the while loop that can change the value of
  $\overline\alpha_i$ is \ref{alg:update-cert}.
  We prove~\eqref{eq:loop-cover-induction} by induction over the while-loop iterations,
  considering only these relevant lines.

  Suppose \eqref{eq:loop-cover-induction} holds at the beginning of an arbitrary iteration, and let
  $a$ be the value of $\overline\alpha_i$ at that time.
  Lemma~\ref{lem:current-pruning-preserves-cover} gives
  $\widehat D_i(a)\subseteq\bigcup_{B\in\mathcal O}B$ after line~\ref{alg:positive-prune}
  when $a>0$ and after line~\ref{alg:threshold-prune} if $a=0$.
  Consider line~\ref{alg:domain-prune-update} and an
  arbitrary $\policy\in\widehat D_i(a)$.
  Since $\widehat D_i(a)\subseteq\bigcup_{B\in\mathcal O}B$ before this line,
  there exists $B\in\mathcal O$ with $\policy\in B$.
  Since $\policy\in\widehat D_i(a)$, by definition, $\policy\in D_{i-1}$, so
  $\underline{\rrrule}_j(B)\leq\rrrule_j(\policy)\leq\overline\alpha_j$ for all $j<i$.
  Thus, line~\ref{alg:domain-prune-update} retains $B$
  and $\widehat D_i(a)\subseteq\bigcup_{B\in\mathcal O}B$ immediately before
  line~\ref{alg:update-cert}.

  Finally, line~\ref{alg:update-cert} cannot increase $\overline\alpha_i$
  because $\mathcal C$ only grows in the while loop in line~\ref{alg:add-incumbents}.
  Thus, after line~\ref{alg:update-cert},
  $\widehat D_i(\overline\alpha_i) \subseteq \widehat D_i(a)$, and
  \eqref{eq:loop-cover-induction} holds at the beginning of the next iteration.
\end{proof}

\begin{lemma}
  \label{lem:certificate-update-invariants}
  Immediately before every execution of \textsc{InitCert} and \textsc{UpdateCert} during level $i$,
  \begin{equation}
    \emptyset\neq\mathcal C\subseteq D_{i-1}
    \quad\text{and}\quad
    M_i\subseteq\bigcup_{B\in\mathcal O}B .
    \label{eq:cert-update-hypotheses}
  \end{equation}
\end{lemma}
\begin{proof}
  We first prove $\emptyset\neq\mathcal C\subseteq D_{i-1}$ immediately before lines
  \ref{alg:init-cert} and \ref{alg:update-cert}.
  This holds before line~\ref{alg:init-cert} by assumption.
  During the while loop, the only line that changes $\mathcal C$ before line~\ref{alg:update-cert}
  is line~\ref{alg:add-incumbents}, which adds the finite set
  $\mathcal S\subseteq D_{i-1}$ returned by \textsc{CSamples}.
  Hence, $\emptyset\neq\mathcal C\subseteq D_{i-1}$ immediately before every execution of
  lines~\ref{alg:init-cert} and~\ref{alg:update-cert}.

  We next prove $M_i\subseteq\bigcup_{B\in\mathcal O}B$ immediately before line~\ref{alg:init-cert}.
  From the proof of Lemma~\ref{lem:cover-invariant},
  $D_{i-1}\subseteq\bigcup_{B\in\mathcal O}B$ holds after
  line~\ref{alg:domain-prune-init}, and hence
  $M_i\subseteq D_{i-1}\subseteq\bigcup_{B\in\mathcal O}B$ immediately before line~\ref{alg:init-cert}.

  It remains to prove $M_i\subseteq\bigcup_{B\in\mathcal O}B$ immediately before line~\ref{alg:update-cert}.
  From the proof of Lemma~\ref{lem:cover-invariant}, right before line~\ref{alg:update-cert},
  $\widehat D_i(\overline\alpha_i)\subseteq\bigcup_{B\in\mathcal O}B$.
  Lemma~\ref{lem:Mi}, applied immediately before line~\ref{alg:update-cert}, gives
  $M_i\subseteq\widehat D_i(\overline\alpha_i)\subseteq\bigcup_{B\in\mathcal O}B$
  immediately before line~\ref{alg:update-cert}.
\end{proof}

\begin{lemma}
  \label{lem:level-invariants}
  If the while loop terminates, then
  at the end of level $i$ (i.e., immediately after line~\ref{alg:filter-incumbents}),
  \begin{enumerate}[label=(\roman*)]
  \item $D_i$ is compact;
  \item $D_i\subseteq\bigcup_{B\in\mathcal O}B$;
  \item $\overline\alpha_i-\alpha_i(D_{i-1})\leq\varepsilon_i$;
  \item $\mathcal C$ is a nonempty finite subset of $D_i$.
  \end{enumerate}
\end{lemma}

\begin{proof}
  To prove (i), we note that the only line at level $i$ that assigns a value to $D_i$ is
  line~\ref{alg:set-domain}, which simply sets
  $D_i=\widehat D_i(\overline\alpha_i)$.
  Since $\rrrule_i$ is continuous,
  $\widehat D_i(\overline\alpha_i)$ is a closed subset of $D_{i-1}$.
  Since $D_{i-1}$ is compact, $D_i=\widehat D_i(\overline\alpha_i)$ is compact.

  To prove (ii), from Lemma~\ref{lem:cover-invariant}, the
  inclusion \eqref{eq:loop-cover-induction} holds at the beginning of each iteration of the while loop.
  If the computational budget is exhausted before an iteration begins or immediately after
  line~\ref{alg:update-cert},
  the same inclusion holds immediately before line~\ref{alg:set-domain}.
  If line~\ref{alg:threshold-stop} terminates the loop, then $\overline\alpha_i=0$.
  From Lemma~\ref{lem:cover-invariant}, the
  inclusion \eqref{eq:loop-cover-induction} holds immediately before
  line~\ref{alg:if-not-zero} and
  Lemma~\ref{lem:current-pruning-preserves-cover} gives the same inclusion after
  line~\ref{alg:threshold-prune}.
  Therefore, in all cases, immediately before line~\ref{alg:set-domain},
  the inclusion \eqref{eq:loop-cover-induction} holds.
  Since line~\ref{alg:set-domain} sets $D_i=\widehat D_i(\overline\alpha_i)$, we have
  $D_i\subseteq\bigcup_{B\in\mathcal O}B$.

  To prove (iii), Lemma~\ref{lem:certificate-update-invariants} verifies the hypotheses of
  Lemma~\ref{lem:gap-certificate} at every execution of \textsc{InitCert} and \textsc{UpdateCert}.
  Therefore the final certificate values at level $i$ satisfy
  $\underline\alpha_i\leq\alpha_i(D_{i-1})$.
  Subtracting $\overline\alpha_i$ from this inequality and replacing
  $\varepsilon_i=\overline\alpha_i-\underline\alpha_i$ give
  $\overline\alpha_i-\alpha_i(D_{i-1})\leq\varepsilon_i$.

  Finally, we prove (iv).
  Immediately before line~\ref{alg:filter-incumbents},
  $\mathcal C$ is finite and satisfies $\mathcal C\subseteq D_{i-1}$,
  because by assumption, it is a nonempty finite subset of $D_{i-1}$ at the beginning of the level,
  and the only line that modifies it before line~\ref{alg:filter-incumbents} is
  line~\ref{alg:add-incumbents}, which adds the finite set
  $\mathcal S\subseteq D_{i-1}$ returned by \textsc{CSamples}.
  The set $\mathcal C$ is also nonempty because it is nonempty at the beginning of the level and the
  algorithm only adds policies before line~\ref{alg:filter-incumbents}.
  Choose $\policy^{\rm inc}\in\argmin_{\policy\in\mathcal C}\rrrule_i(\policy)$.
  By~\eqref{eq:anytime-gap},
  $\rrrule_i(\policy^{\rm inc})=\overline\alpha_i$.
  Since $\policy^{\rm inc}\in D_{i-1}$, line~\ref{alg:set-domain} implies
  $\policy^{\rm inc}\in D_i$.
  Therefore $\mathcal C\cap D_i$ is nonempty.
  Line~\ref{alg:filter-incumbents} replaces $\mathcal C$ by $\mathcal C\cap D_i$, so after that line
  $\mathcal C$ is a nonempty finite subset of $D_i$.
\end{proof}

We now use the invariants established in Lemma~\ref{lem:level-invariants}
to prove the correctness of Algorithm~\ref{alg:anytime-rulebook}.

\begin{theorem}
  \label{thm:algorithm-correctness}
  Suppose Assumptions~\ref{ass:compact} and \ref{ass:lipschitz} hold, and each rule level terminates.
  Then, Algorithm~\ref{alg:anytime-rulebook} returns a policy $\widehat\policy$ that
  satisfies
  $\rrrule_i(\widehat\policy)-\alpha_i(D_{i-1})\leq\varepsilon_i$,
  for all $i \in \{1, \ldots, m\}$,
  i.e., $\varepsilon_i$ is a valid finite-budget optimality certificate.
\end{theorem}

\begin{proof}
  We first show that the hypotheses of Lemma~\ref{lem:level-invariants} hold at the beginning of
  every rule level.  The proof is by induction on the level $i$.
  For $i=1$, line~\ref{alg:init-domain} sets $D_0=\policies$, which is compact by
  Assumption~\ref{ass:compact}.
  Line~\ref{alg:init-cover} initializes
  $\mathcal O$ so that $D_0\subseteq\bigcup_{B\in\mathcal O}B$, and
  line~\ref{alg:init-incumbents} initializes $\mathcal C$ as a nonempty finite subset of $D_0$.
  Hence Lemma~\ref{lem:level-invariants} applies at level $1$.

  Now suppose the hypotheses of Lemma~\ref{lem:level-invariants} hold at the beginning of
  an arbitrary level $i$.
  Since level $i$ terminates by assumption, Lemma~\ref{lem:level-invariants} implies, immediately after
  line~\ref{alg:filter-incumbents}, that $D_i$ is compact,
  $D_i\subseteq\bigcup_{B\in\mathcal O}B$, and $\mathcal C$ is a nonempty finite subset of $D_i$.
  Thus, the hypotheses of Lemma~\ref{lem:level-invariants} hold at the beginning of level $i+1$.
  By induction, they hold at the beginning of every rule level.

  Applying Lemma~\ref{lem:level-invariants}(iii) at each level gives
  \begin{equation}
    \overline\alpha_i-\alpha_i(D_{i-1})\leq\varepsilon_i,
    \qquad i=1,\ldots,m.
    \label{eq:algorithm-gap-correctness}
  \end{equation}
  Lemma~\ref{lem:level-invariants}(iv), applied at level $m$, gives
  $\emptyset\neq\mathcal C\subseteq D_m$ after line~\ref{alg:filter-incumbents}.
  Therefore line~\ref{alg:return} is well-defined and returns some
  $\widehat\policy\in\mathcal C\subseteq D_m$.
  Additionally, by definition of $D_{i}$ in \eqref{eq:certified-domain},
  $D_{i} \subseteq D_{i-1}$ for all rule level $i$.
  Hence $D_m\subseteq D_i$ for every $i\leq m$.
  Since $\widehat\policy\in D_m$, we have $\widehat\policy\in D_i$ for every $i$.
  Using the definition of $D_i$ again gives
  $\rrrule_i(\widehat\policy)\leq\overline\alpha_i$ for every $i$.
  Combining this with~\eqref{eq:algorithm-gap-correctness} yields
  $\rrrule_i(\widehat\policy)-\alpha_i(D_{i-1}) \leq \overline\alpha_i-\alpha_i(D_{i-1})
  \leq\varepsilon_i$ for all $i \in \{1,\ldots,m\}$.
\end{proof}

\subsection{Convergence}
The previous subsection proves that Algorithm~\ref{alg:anytime-rulebook} returns a policy with a
certified optimality gap $\varepsilon_{i}$ at every rule level $i$.
This subsection studies the asymptotic behavior of the algorithm as the computational budget increases.
Throughout this subsection, fix an arbitrary rule level $i\in\{1,\ldots,m\}$ and assume the
hypotheses of Lemma~\ref{lem:level-invariants} hold.

We index the reachable snapshots of the while loop by the number
$N\in\naturals$ of completed full iterations of the while loop at level $i$.
Snapshot $N$ is taken in the $(N+1)$-st iteration of the while loop,
immediately after the construction of $\mathcal R$, namely
after line~\ref{alg:set-refine-positive} when $\overline\alpha_i>0$,
and after line~\ref{alg:threshold-refine} when $\overline\alpha_i=0$.
The incumbent set and certificate values at snapshot $N$
are denoted by $\mathcal C(N)$, $\underline\alpha_i(N)$, $\overline\alpha_i(N)$, and
$\varepsilon_i(N)$.
These certificates are those computed by \textsc{InitCert} (line~\ref{alg:init-cert}) when $N=0$
and by \textsc{UpdateCert} (line~\ref{alg:update-cert})
at the end of the $N$-th iteration when $N\ge1$ before the subsequent pruning step.
Let $\mathcal O^{-}(N)$ denote the corresponding value of $\mathcal O$ at this certificate-update
line.
The subsequent pruning step uses $\overline\alpha_i(N)$ to construct the box collections
$\mathcal O(N)$ and $\mathcal R(N)$.
Since the pruning step does not modify $\mathcal C(N)$ or recompute the certificates,
the pre-pruning quantities $\mathcal C(N)$, $\underline\alpha_i(N)$, $\overline\alpha_i(N)$, and
$\varepsilon_i(N)$ and the post-pruning collections $\mathcal O(N)$ and $\mathcal R(N)$ all belong
to the same snapshot $N$.

The set of boxes selected for possible refinement is
\begin{equation}
  \mathcal R(N)=
  \begin{cases}
    \mathcal O(N), & \overline\alpha_i(N)>0,\\
    \{B\in\mathcal O(N)\mid \overline q_i(B)>\threshold_i\},
                   & \overline\alpha_i(N)=0.
  \end{cases}
  \label{eq:refinable-boxes}
\end{equation}
If line~\ref{alg:threshold-stop} terminates level $i$,
this final snapshot has $\mathcal R(N)=\emptyset$.
Define the refinement width at snapshot $N$ by
\begin{equation}
  \delta_i(N)=
  \begin{cases}
    \max_{B\in\mathcal R(N)}w_i(B), & \mathcal R(N)\neq\emptyset,\\
    0, & \mathcal R(N)=\emptyset.
  \end{cases}
  \label{eq:refinement-width}
\end{equation}

The first two results show that longest-side refinement drives this refinement width to zero.

\begin{lemma}
  \label{lem:monotonic}
  The sequence $\{\delta_i(N)\}_{N\geq0}$ is monotonically nonincreasing.
\end{lemma}
\begin{proof}
  If $L_i=0$, then \eqref{eq:LBUB} gives $w_i(B)=0$ for every box $B$,
  and the result is immediate.

  Suppose $L_i>0$.
  Since $w_i(B)=L_i\sqrt d\,l(B)$ and $L_{i}$ and $d$ are constant,
  $w_{i}(B)$ is proportional to the longest side length $l(B)$.
  Consider an arbitrary $j<N$.
  Since the next bisection is performed, $\mathcal R(j)\neq\emptyset$.
  Line~\ref{alg:select-box} selects
  $B^\star\in\argmax_{B\in\mathcal R(j)}w_i(B)$, so $w_i(B^\star)=\delta_i(j)$.
  Let $B^+$ and $B^-$ be the children created by line~\ref{alg:split-box}.
  Longest-side bisection gives $l(B^+),l(B^-)\leq l(B^\star)$. Hence,
  $w_i(B^+),w_i(B^-)\leq w_i(B^\star)=\delta_i(j)$.

  We want to show that $w_i(B)\leq\delta_i(j)$ for every $B\in\mathcal R(j+1)$.
  By construction,
  $\mathcal R(j+1) \subseteq (\mathcal R(j) \setminus \{B^{*}\}) \cup \{B^{+}, B^{-}\}$
  because lines~\ref{alg:domain-prune-update},
  \ref{alg:positive-prune}--\ref{alg:set-refine-positive},
  and \ref{alg:threshold-prune}--\ref{alg:threshold-refine} only remove boxes or assign
  $\mathcal R$ from the boxes that remain after pruning.
  Since $w_i(B^+),w_i(B^-)\leq\delta_i(j)$ and any $B \in \mathcal R(j) \setminus \{B^{*}\}$
  is unchanged from snapshot $j$ to $j+1$,
  we can conclude that $w_i(B)\leq\delta_i(j)$ for every $B\in\mathcal R(j+1)$.
  Therefore, if $\mathcal R(j+1)\not=\emptyset$, then
  $\delta_i(j+1)=\max_{B\in\mathcal R(j+1)}w_i(B)\leq\delta_i(j)$.
  Otherwise, $\delta_i(j+1)=0\leq\delta_i(j)$.
\end{proof}

\begin{theorem}
  \label{thm:delta-zero}
  $\delta_i(N)\to0$.
\end{theorem}
\begin{proof}
  If $L_i=0$, then $\delta_i(N)=0$ for all $N$.
  Suppose $L_i>0$.
  By Lemma~\ref{lem:monotonic}, $\{\delta_i(N)\}_{N\geq0}$ is monotonically nonincreasing and
  bounded below by zero.
  Hence, by the monotone convergence theorem, there exists
  $\delta_\infty\geq0$ such that
  $\delta_i(N)\to\delta_\infty$.
  Suppose, for contradiction, that $\delta_\infty>0$, and set
  $\varepsilon=\delta_\infty/2$ and $\eta=\varepsilon/(L_i\sqrt d)$.

  Let $\mathcal F_\eta=\{B \in \bigcup_{N \in \naturals} \mathcal O(N) \mid l(B)\geq\eta\}$
  be the set of generated boxes whose longest side length is at least $\eta$.
  We first show that $\mathcal F_\eta$ is finite.
  Fix $B^0\in\mathcal O(0)$ and choose $q\in\naturals$ such that $l(B^0)/2^q<\eta$.
  Consider an arbitrary sequence $B^0,B^1,\ldots,B^H$ in which $B^{j+1}$ is one of the two boxes obtained
  from $B^j$ by line~\ref{alg:split-box}.
  It is enough to bound $H$ such that $l(B^H)\geq\eta$.

  Let $c_j\in\{1,\ldots,d\}$ be the coordinate bisected when passing from $B^j$ to $B^{j+1}$.
  Since bisection does not increase any side length,
  $l(B^0)\geq l(B^1)\geq\cdots\geq l(B^H)$.
  Thus, if $l(B^H)\geq\eta$, then for every $j\in\{0,\ldots,H-1\}$,
  \begin{equation}
    l_{c_j}(B^j)=l(B^j)\geq\eta>l(B^0)/2^q,
    \label{eq:thm:delta-zero-inequality}
  \end{equation}
  where the equality follows from line~\ref{alg:split-box}, which always bisects
  a longest side of the box.

  Fix a coordinate $c\in\{1,\ldots,d\}$, and let
  $j_1<\cdots<j_h$ be all indices $j\in\{0,\ldots,H-1\}$ such that $c_j=c$.
  Then,
  \begin{equation}
    l_{c}(B^{j_{s}}) = l_c(B^0)/2^{s-1}, \forall s \in \{1, \ldots, h\}
    \label{eq:thm:delta-zero-equality}
  \end{equation}
  because only previous splits with $c_j=c$ change this
  side length.
  Additionally, setting $j = j_{s}$ in \eqref{eq:thm:delta-zero-inequality} gives
  $l_{c}(B^{j_{s}}) > l(B^{0})/2^{q}$.
  Combining this with \eqref{eq:thm:delta-zero-equality} gives
  \[
    \frac{l_c(B^0)}{2^{s-1}} > \frac{l(B^0)}{2^q},
    \qquad \forall s \in \{1,\ldots,h\}.
  \]
  Since $l_c(B^0)\leq l(B^0)$, this is possible only if $s\leq q$.
  Hence $h\leq q$ for every coordinate $c$.
  Summing over coordinates gives
  $H=\sum_{c=1}^d |\{j\in\{0,\ldots,H-1\}\mid c_j=c\}|\leq dq$.
  Therefore, for this fixed $B^0$, every generated box $B$ with $l(B)\geq\eta$ is obtained after at
  most $dq$ bisections from $B^0$.
  There are at most $\sum_{k=0}^{dq}2^k$ such boxes.
  Since $\mathcal O(0)$ is finite, the total number of generated boxes $B$ with $l(B)\geq\eta$ is
  finite.

  Since $\delta_i(N)\to\delta_\infty$ and $\varepsilon=\delta_\infty/2 > 0$,
  based on the definition of convergence, there exists $N_0$ such that
  \begin{equation}
    \delta_i(N)\geq\varepsilon>0,
    \qquad \forall N\geq N_0.
    \label{eq:delta-tail-lower-bound}
  \end{equation}
  Since $\delta_{i}(N)>0$, it follows from \eqref{eq:refinement-width} that
  $\mathcal R(N)\neq\emptyset$ for all $N \geq N_{0}$.
  Hence line~\ref{alg:select-box} selects a box
  $B^\star_N\in\mathcal R(N)$ with $w_i(B^\star_N)=\delta_i(N)$.
  Using $w_i(B)=L_i\sqrt d\,l(B)$ and \eqref{eq:delta-tail-lower-bound},
  \[
    l(B^\star_N)=\frac{\delta_i(N)}{L_i\sqrt d}
    \geq \frac{\varepsilon}{L_i\sqrt d}=\eta.
  \]
  Thus
  \begin{equation}
    B^\star_N\in\mathcal F_\eta,
    \qquad \forall N\geq N_0.
    \label{eq:selected-large-boxes}
  \end{equation}
  Because there are infinitely many indices $N\geq N_0$, \eqref{eq:selected-large-boxes} requires
  infinitely many selections from $\mathcal F_\eta$.
  But whenever a generated box $B$ is selected at line~\ref{alg:select-box},
  lines~\ref{alg:split-box}--\ref{alg:update-cover} remove $B$ from $\mathcal O$ and replace it by
  its two children.
  Hence the same generated box cannot be selected at any later snapshot.
  Therefore \eqref{eq:selected-large-boxes} requires infinitely many distinct boxes in
  $\mathcal F_\eta$, contradicting the fact that $\mathcal F_\eta$ is finite.
  This contradiction shows that $\delta_\infty=0$.
\end{proof}

We now analyze the asymptotic behavior of the certified optimality gap $\varepsilon_i(N)$
and the policies returned by the algorithm.
We first prove the convergence of $\varepsilon_i(N)$.

\begin{lemma}
  \label{lem:monotone-running-gap}
  The sequence $\{\varepsilon_i(N)\}_{N\geq0}$ is monotonically nonincreasing and converges to
  a limit $\varepsilon_{i,\infty}\geq0$.
\end{lemma}
\begin{proof}
  During a fixed rule level, $\mathcal C(N)$ only grows, so
  $\overline\alpha_i(N)=\min_{\policy\in\mathcal C(N)}\rrrule_i(\policy)$ is monotonically
  nonincreasing.
  By the update rule in~\eqref{eq:cert-update}, $\underline\alpha_i(N)$ is monotonically
  nondecreasing.
  Hence $\varepsilon_i(N)=\overline\alpha_i(N)-\underline\alpha_i(N)$ is monotonically
  nonincreasing.
  Lemma~\ref{lem:gap-certificate} and Lemma~\ref{lem:certificate-update-invariants} give
  $\underline\alpha_i(N)\leq\alpha_i(D_{i-1})\leq\overline\alpha_i(N)$, so
  $\varepsilon_i(N)\geq0$.
  By the monotone convergence theorem,
  $\{\varepsilon_i(N)\}_{N\geq0}$ converges.
\end{proof}

The preceding lemma does not characterize the limiting value $\varepsilon_{i,\infty}$.
To show that $\varepsilon_i(N)$ converges to zero,
we additionally assume that line~\ref{alg:init-cover} initializes $\mathcal O$ as an exact
finite union of closed axis-aligned boxes representing $\policies$.
Additionally, we assume that the incumbent set contains a
certified policy whose value approaches the optimum over $D_{i-1}$
at the rate of the refinement width.

\begin{assumption}
  \label{ass:exact}
  At line~\ref{alg:init-cover}, $\mathcal O$ is initialized as a finite set of closed
  axis-aligned boxes satisfying
  $\policies = \bigcup_{B\in\mathcal O}B$.
\end{assumption}

\begin{assumption}
  \label{ass:complete-box-processing}
  There exists $\kappa_i\geq0$ such that,
  for every reachable snapshot $N$, there is a policy
  $\policy_N\in\mathcal C(N)\cap D_{i-1}$ satisfying
  \begin{equation}
    \rrrule_i(\policy_N)
    \leq
    \alpha_i(D_{i-1})+\kappa_i\delta_i(N).
    \label{eq:near-optimal-incumbent}
  \end{equation}
\end{assumption}

Assumption~\ref{ass:exact} is satisfied whenever $\policies$ can be represented as a
finite union of closed axis-aligned boxes.
Assumption~\ref{ass:complete-box-processing} requires
the incumbent set to contain a policy whose excess risk converges to the optimum at the same
rate as the box resolution.
It is satisfied, e.g., if \textsc{CSamples} generates a sufficiently dense set of
certified policies in $D_{i-1}$, or
if it generates one certified policy from each box in the current box collection.
Since $M_i\subseteq\bigcup_{B\in\mathcal O(N)}B$, at least one box in the current box collection
contains a minimizer. At snapshots with $\overline\alpha_i(N)>0$, \eqref{eq:refinable-boxes}
gives $\mathcal R(N)=\mathcal O(N)$, so every box in $\mathcal O(N)$ has width at most
$\delta_i(N)$. Lipschitz continuity then implies that a certified policy sampled from such a box
has excess risk within
$O(\delta_i(N))$ of $\alpha_i(D_{i-1})$.

\begin{lemma}
  \label{lem:box-containment}
  Under Assumption~\ref{ass:exact},
  every box generated by Algorithm~\ref{alg:anytime-rulebook} satisfies
  $B\subseteq\policies$ and $c(B) \in \policies$.
\end{lemma}
\begin{proof}
  By Assumption~\ref{ass:exact}, the initial box collection satisfies
  $\policies = \bigcup_{B\in\mathcal O}B$.
  Thus, every initial box $B\in\mathcal O$ satisfies $B\subseteq\policies$.
  Algorithm~\ref{alg:anytime-rulebook} refines boxes only by axis-aligned bisection,
  which replaces a box by subsets of that box.
  Thus, every generated box remains contained in
  its initial ancestor and therefore in $\policies$.
  Hence, $B\subseteq\policies$ for every generated box $B$.
  Because axis-aligned boxes are convex, their centers satisfies
  $c(B)\in B\subseteq\policies$.
\end{proof}

By Lemma~\ref{lem:box-containment} and Lemma~\ref{lem:lipschitz-bounds},
the center-based construction~\eqref{eq:LBUB} is a valid implementation of
\textsc{Bounds}.
We next show that centers of boxes in $\mathcal O(N)$ become arbitrarily close to the certified admissible domain
as the box size decreases.
For any nonempty compact set $S\subseteq\policies$ and any $\policy\in\policies$, define
$\operatorname{dist}(\policy,S)=\min_{\policy'\in S}\|\policy-\policy'\|_2$.

\begin{lemma}
  \label{lem:domain-consistent-boxes}
  Suppose Assumption~\ref{ass:exact} holds.
  Let $\{N_k\}_{k\geq1}\subseteq\naturals$ be a sequence of reachable snapshot indices with
  $N_k\to\infty$.
  For each $k$, choose one box $B_k\in\mathcal O(N_k)$.
  If $l(B_k)\to0$, then $\operatorname{dist}(c(B_k),D_{i-1})\to0$.
\end{lemma}
\begin{proof}
  By Lemma~\ref{lem:box-containment},
  $B_k\subseteq\policies$ and $c(B_k)\in\policies$ for every $k$.
  If $i=1$, then $D_0=\policies$, so
  $\operatorname{dist}(c(B_k),D_0)=0$ for every $k$.

  Suppose $i>1$.
  Consider a sequence $\{B_{k}\}$ with $l(B_{k})\to 0$.
  We first show that every convergent subsequence of $\{c(B_k)\}$ has its limit in $D_{i-1}$.
  Consider an arbitrary subsequence of $\{c(B_k)\}$, denote its indices by
  $k_{1} < k_{2} < \ldots$, and let its limit be $\policy^0$.
  In other words, $c(B_{k_\ell})\to\policy^0$ as $l \to \infty$.
  Since $c(B_{k_\ell})\in\policies$ for every $\ell$ and $\policies$ is compact,
  $\policy^0\in\policies$.

  Now fix an arbitrary $j<i$.
  Since
  $B_{k_\ell}\in\mathcal O(N_{k_\ell})\subseteq\mathcal O^{-}(N_{k_\ell})$,
  the box $B_{k_\ell}$ has already been processed by \textsc{DPrune}
  (line~\ref{alg:domain-prune-init} when $N_{k_\ell}=0$ and
  line~\ref{alg:domain-prune-update} when $N_{k_\ell}\ge1$).
  Hence, $\underline{\rrrule}_j(B_{k_\ell})\le\overline\alpha_j$.
  Combining this with \eqref{eq:v-bounds} and \eqref{eq:v-width} gives
  \begin{align*}
    \rrrule_j(c(B_{k_\ell}))
    &\leq
    \overline{\rrrule}_{j}(B_{k_{\ell}}) \\
    &\leq
    \underline{\rrrule}_j(B_{k_\ell})+w_{j}(B_{k_\ell}) \\
    &\leq
    \overline\alpha_j+w_{j}(B_{k_\ell}).
  \end{align*}
  Since $l(B_{k_\ell})\to0$,
  $w_{j}(B_{k_{\ell}}) = L_j\sqrt d\,l(B_{k_\ell}) \to 0$.
  Moreover,
  since $c(B_{k_\ell})\to\policy^0$ and $\rrrule_j$ is continuous,
  $\rrrule_j(c(B_{k_\ell}))\to\rrrule_j(\policy^0)$.
  Taking limits in the inequality above therefore yields
  $\rrrule_j(\policy^0)\leq\overline\alpha_j$.
  Since $j<i$ was arbitrary, we can conclude that
  $\rrrule_j(\policy^0)\leq\overline\alpha_j$ for all $j<i$.
  Therefore
  \[
    \policy^0\in
    \{\policy\in\policies\mid
    \rrrule_j(\policy)\leq\overline\alpha_j,\ \forall j \in \{1,\ldots,i-1\}\}
    =D_{i-1}.
  \]

  Suppose, for contradiction, that
  $\operatorname{dist}(c(B_k),D_{i-1})\not\to0$.
  Then there exist $\epsilon>0$ and a subsequence, which we relabel as
  $\{B_k\}$, such that
  $\operatorname{dist}(c(B_k),D_{i-1})\geq\epsilon$ for every $k$.
  Since every $c(B_{k}) \in \policies$ for every $k$ and $\policies$ is compact,
  the Bolzano--Weierstrass theorem implies that $\{c(B_{k})\}$ has a convergent subsequence
  $\{c(B_{k_\ell})\}$.
  By the preceding paragraph, $c(B_{k_\ell}) \to \policy^0$ for some $\policy^0\in D_{i-1}$.
  Thus, $\operatorname{dist}(c(B_{k_\ell}), D_{i-1}) \leq \|c(B_{k_\ell}) - \policy^0\|_2 \to 0$.
  But this contradicts $\operatorname{dist}(c(B_{k_\ell}),D_{i-1})\geq\epsilon$ for every $\ell$.
\end{proof}

\begin{theorem}
  \label{thm:algorithmic-gap-convergence}
  Under Assumptions~\ref{ass:exact} and \ref{ass:complete-box-processing},
  $\varepsilon_i(N)\to0$.
\end{theorem}
\begin{proof}
  Let $a_i^\star=\alpha_i(D_{i-1})$.
  Lemma~\ref{lem:gap-certificate} and Lemma~\ref{lem:certificate-update-invariants} give,
  for every reachable snapshot $N$,
  \begin{equation}
    \underline\alpha_i(N)\leq a_i^\star\leq\overline\alpha_i(N).
    \label{eq:bounds-alpha}
  \end{equation}

  We first prove the result when $L_i=0$.
  In this case, $q_i$ and $\rrrule_i$ are constant on $\policies$.
  Since $\emptyset\neq\mathcal C(N)\subseteq D_{i-1}$, we have
  $\overline\alpha_i(N)=a_i^\star$.
  Also, $\underline{\rrrule}_i(B)=a_i^\star$ for every $B\in\mathcal O^{-}(N)$, so
  \eqref{eq:anytime-gap}--\eqref{eq:cert-update} give
  $\underline\alpha_i(N)=a_i^\star$.
  Hence $\varepsilon_i(N)=0$ for every $N$.

  Assume $L_i>0$.
  Since $\overline\alpha_i(N)$ is nonnegative and nonincreasing, either
  $\overline\alpha_i(N)=0$ eventually or $\overline\alpha_i(N)>0$ for every $N$.
  We next prove the result when $\overline\alpha_i(N)=0$ eventually.
  By \eqref{eq:bounds-alpha}, $a_i^\star=0$ eventually.
  Since $\underline{\rrrule}_i(B)\geq0$ and \eqref{eq:anytime-gap}--\eqref{eq:cert-update} preserve
  $\underline\alpha_i(N)\geq0$, \eqref{eq:bounds-alpha} gives
  $\underline\alpha_i(N)=0$ eventually.
  Hence $\varepsilon_i(N)=0$ eventually.

  It remains to prove the result when $L_i>0$ and $\overline\alpha_i(N)>0$ for every $N$.
  Then \eqref{eq:refinable-boxes} gives $\mathcal R(N)=\mathcal O(N)$.
  Assumption~\ref{ass:complete-box-processing} gives
  $\policy_N\in\mathcal C(N)\cap D_{i-1}$ such that
  $\rrrule_i(\policy_N)\leq a_i^\star+\kappa_i\delta_i(N)$.
  Since $\policy_N\in\mathcal C(N)$,
  $\overline\alpha_i(N)\leq\rrrule_i(\policy_N)$.
  Therefore
  \begin{equation}
    0\leq \overline\alpha_i(N)-a_i^\star
    \leq \kappa_i\delta_i(N).
    \label{eq:upper-converges-alpha}
  \end{equation}

  We now prove that $\underline\alpha_i(N)\to a_i^\star$.
  Since $\mathcal R(N)=\mathcal O(N)$, \eqref{eq:refinement-width} gives
  $w_i(B)\leq\delta_i(N)$ for every $B\in\mathcal O(N)$.
  By Theorem~\ref{thm:delta-zero}, $\delta_i(N)\to0$,
  and since $w_i(B)=L_i\sqrt d\,l(B)$, we have
  $\max_{B\in\mathcal O(N)}l(B)\to0$.
  Together with Lemma~\ref{lem:domain-consistent-boxes}, this implies
  $\max_{B\in\mathcal O(N)} \operatorname{dist}(c(B),D_{i-1})\to0$.
  Thus,
  $\eta_N=
  \max_{B\in\mathcal O(N)}
  \bigl(\operatorname{dist}(c(B),D_{i-1})+\radius(B)\bigr)
  \to 0$.

  Fix $B\in\mathcal O(N)$ and choose
  $\policy_B\in D_{i-1}$ such that
  $\|c(B)-\policy_B\|_2=\operatorname{dist}(c(B),D_{i-1})$.
  Then,
  \begin{align*}
    \underline{\rrrule}_i(B)
    &=\max\{q_i(c(B))-L_i\radius(B)-\threshold_i,0\} \\
    &\geq \rrrule_i(c(B))-L_i\radius(B) \\
    &\geq \rrrule_i(\policy_B)-L_i\|c(B)-\policy_B\|_2-L_i\radius(B) \\
    &\geq a_i^\star-L_i\eta_N .
  \end{align*}
  Therefore
  \begin{equation}
    m_N=\min_{B\in\mathcal O(N)}\underline{\rrrule}_i(B)
    \geq a_i^\star-L_i\eta_N .
    \label{eq:post-pruning-lower-bound}
  \end{equation}

  We now relate $m_N$ to $\underline\alpha_i(N)$.
  By the snapshot convention, $\mathcal O^{-}(N)$ is the box collection used in
  \textsc{InitCert} when $N=0$ and in \textsc{UpdateCert} when $N\geq1$.
  Thus
  \begin{equation}
    \underline\alpha_i(N)
    \geq
    \min_{B\in\mathcal O^{-}(N)}\underline{\rrrule}_i(B).
    \label{eq:alpha-lower-preprune}
  \end{equation}
  Since $\overline\alpha_i(N)>0$, line~\ref{alg:positive-prune} produces
  $\mathcal O(N)$ from $\mathcal O^{-}(N)$ by removing exactly the boxes satisfying
  $\underline{\rrrule}_i(B)>\overline\alpha_i(N)$.
  Lemma~\ref{lem:Mi} gives
  $M_i\subseteq\widehat D_i(\overline\alpha_i(N))$.
  Lemma~\ref{lem:cover-invariant} and Lemma~\ref{lem:current-pruning-preserves-cover} give
  $\widehat D_i(\overline\alpha_i(N))\subseteq\bigcup_{B\in\mathcal O(N)}B$.
  Hence $M_i\subseteq\bigcup_{B\in\mathcal O(N)}B$.
  Therefore there exist $\policy^m\in M_i$ and $B^m\in\mathcal O(N)$ such that
  $\policy^m\in B^m$.
  By \eqref{eq:v-bounds} and \eqref{eq:bounds-alpha},
  $m_N\leq\underline{\rrrule}_i(B^m)
  \leq \rrrule_i(\policy^m)
  =a_i^\star
  \leq\overline\alpha_i(N)$.
  Therefore every box removed from $\mathcal O^{-}(N)$ has lower bound strictly larger than $m_N$,
  and
  \begin{equation}
    \min_{B\in\mathcal O^{-}(N)}\underline{\rrrule}_i(B)
    =
    \min_{B\in\mathcal O(N)}\underline{\rrrule}_i(B)
    =m_N .
    \label{eq:pre-post-min-equal}
  \end{equation}
  Combining \eqref{eq:post-pruning-lower-bound}, \eqref{eq:alpha-lower-preprune}, and
  \eqref{eq:pre-post-min-equal} gives
  \begin{equation}
    \underline\alpha_i(N)\geq a_i^\star-L_i\eta_N .
    \label{eq:lower-converges-alpha}
  \end{equation}

  Finally, \eqref{eq:bounds-alpha}, \eqref{eq:upper-converges-alpha}, and
  \eqref{eq:lower-converges-alpha} imply
  \[
    0\leq\varepsilon_i(N)
    =\overline\alpha_i(N)-\underline\alpha_i(N)
    \leq \kappa_i\delta_i(N)+L_i\eta_N.
  \]
  Since $\delta_i(N)\to0$ and $\eta_N\to0$, $\varepsilon_i(N)\to0$.
\end{proof}

\begin{theorem}
  \label{thm:returned-incumbent-convergence}
  For every reachable snapshot $N$, let
  $\widehat\policy_N\in\argmin_{\policy\in\mathcal C(N)}\rrrule_i(\policy)$.
  Under Assumptions~\ref{ass:exact} and \ref{ass:complete-box-processing}, if
  $\{N_k\}_{k\geq1}$ is a sequence of reachable snapshot indices with $N_k\to\infty$ and
  $\widehat\policy_{N_k}\to\widehat\policy$, then $\widehat\policy\in M_i$.
\end{theorem}
\begin{proof}
  By the definition of $\widehat\policy_N$ and $\overline\alpha_i(N)$,
  $\rrrule_i(\widehat\policy_N)=\overline\alpha_i(N)$ for every snapshot $N$.
  By Lemma~\ref{lem:gap-certificate} and
  Lemma~\ref{lem:certificate-update-invariants},
  $0\leq
  \rrrule_i(\widehat\policy_N)-\alpha_i(D_{i-1})
  \leq
  \overline\alpha_i(N)-\underline\alpha_i(N)
  =
  \varepsilon_i(N)$.
  By Theorem~\ref{thm:algorithmic-gap-convergence},
  $\varepsilon_i(N)\to0$.
  Therefore, along the given sequence $N_k\to\infty$,
  $\rrrule_i(\widehat\policy_{N_k})\to\alpha_i(D_{i-1})$.

  It remains to show that the limiting policy $\widehat\policy$
  attains this optimal value.
  Since $\widehat\policy_{N_k}\in\mathcal C(N_k)\subseteq D_{i-1}$ for every $k$ and
  $D_{i-1}$ is compact, the limit satisfies $\widehat\policy\in D_{i-1}$.
  By continuity of $\rrrule_i$,
  $\rrrule_i(\widehat\policy)
  =\lim_{k\to\infty}\rrrule_i(\widehat\policy_{N_k})
  =\alpha_i(D_{i-1})$.
  Since $\widehat\policy \in D_{i-1}$ and achieves the minimum value of
  $\rrrule_i$ over $D_{i-1}$, we can conclude that
  $\widehat\policy\in M_i$.
\end{proof}


\section{Experimental Results}
\label{sec:experiment}
We evaluate Algorithm~\ref{alg:anytime-rulebook} in two complementary settings.
The synthetic benchmark provides a known ground-truth solution for quantitative evaluation,
whereas the highway merging simulation demonstrates the algorithm's practical applicability to
black-box risk-aware control.

\subsection{Synthetic Benchmark}
We first evaluate Algorithm~\ref{alg:anytime-rulebook} on a synthetic single-rule scenario
($m=1$) for which the exact solution is available in closed form,
and compare it against 2 baseline optimizers.
For simplicity, we omit the subscript~1 throughout this subsection and write
$q$, $\rrrule$, $\alpha$, $\threshold$, and $L$ instead of
$q_1$, $\rrrule_1$, $\alpha_1$, $\threshold_1$, and $L_1$.
Because the exact solution is known, this experiment enables direct evaluation of the algorithm's
convergence, certificate quality, and anytime performance.

\noindent
\textbf{Setup.}
The policy space is the box $\policies = [-1,1]^{2}$ and the
risk-evaluation function is the shifted Rosenbrock-type function
\begin{equation}
  q(\policy) = q(x, y)
  = \tfrac{1}{2}\bigl(\tfrac{1}{2} - x\bigr)^{2} + (y - x^{2})^{2}.
  \label{eq:rosenbrock-q}
\end{equation}
This function has a unique minimizer at
$\policy^\star=(0.5,0.25)$, where $q(\policy^\star)=0$.
When $\threshold = 0$, $\alpha(\policies) = 0$,
the rule is satisfied only at $\policy^{\star}$, and
$\policies^{\star}_{1} = \policies_{1} = \{\policy^{\star}\}$
provide an exact solution against which all optimizers can be compared.
Any $\threshold > 0$ enlarges the feasible set to the sublevel set
$\{\policy : q(\policy) \leq \threshold\}$ around $\policy^{\star}$, in which every policy
attains the minimum excess risk $\alpha(\policies) = 0$.

Since $|\partial q/\partial x| = |x - \tfrac{1}{2} - 4x(y-x^{2})| \leq |x - \frac{1}{2}| +
4|x|(|y|+x^{2}) \leq 9.5$ and
$|\partial q/\partial y| = 2|y-x^{2}| \leq 2(|y| + x^{2}) \leq 4$ for all $x,y \in [-1,1]$,
$\|\nabla q(x,y)\|_2 \leq \sqrt{9.5^{2}+4^{2}} \approx 10.31$,
so any $L \geq 10.31$ is a valid Lipschitz constant.
We use the conservative value $L=11$ throughout.

Because $\policies$ is a single closed axis-aligned box,
Assumption~\ref{ass:exact} is satisfied by
initializing the box collection as $\mathcal{O} = \{\policies\}$.
By Lemma~\ref{lem:box-containment} and Lemma~\ref{lem:lipschitz-bounds},
we implement \textsc{Bounds} using the center-based construction in~\eqref{eq:LBUB}
and implement \textsc{CSamples} by returning the centers of newly generated child boxes,
which are automatically certified to be in $\policies$.
The initial incumbent set is
$\mathcal C={c(\policies)}={(0,0)}$.

\noindent
\textbf{Optimizers.}
We compare Algorithm~\ref{alg:anytime-rulebook} against 2 baselines.
The first is a best-first Lipschitz branch-and-bound (B\&B) \cite{Horst:1996:Global}.
It uses the same the box representation as Algorithm~\ref{alg:anytime-rulebook}
but differs in box selection.
Instead of the widest-refinable-box rule on line~\ref{alg:select-box}, it applies the
best-first selection
$B^{\star} \gets \argmin_{B \in \mathcal{O}} \underline\rrrule(B)$.
The selected box is bisected and replaced by its children, while boxes satisfying
$\underline{\rrrule}(B)>\overline\alpha$ are discarded.
This pruning criterion is identical to \textsc{CPrune} on line~\ref{alg:positive-prune} of
Algorithm~\ref{alg:anytime-rulebook}.
We consider two variants of this baseline that differ only in the computation of $\underline\rrrule(B)$.
The \emph{classical} variant treats the excess-risk function $\rrrule$ as an
$L$-Lipschitz function and use
$\underline\rrrule(B) = \rrrule(c(B)) - L\radius(B)$.
The \emph{thresholded} variant instead uses the certified lower bound
$\underline{\rrrule}(B)$ defined in~\eqref{eq:LBUB},
which exploits the threshold structure of $\rrrule$.
Both variants maintain the same incumbent set $\mathcal C$ and the upper certificate
$\overline\alpha$ as in Algorithm~\ref{alg:anytime-rulebook}.

The second baseline is uniform grid search, where
the policy space is discretized into an axis-aligned $n \times n$ regular grid with $n$ equally
spaced values per axis, giving $n^{2}$ grid points.
At iteration $\ell \in \{1, \ldots, n^{2}\}$, the algorithm evaluates $\rrrule$ at
$\policy_{\ell}
= \bigl(-1 + ((\ell{-}1) \bmod n)\,\Delta,\;
-1 + \lfloor (\ell{-}1)/n \rfloor\,\Delta\bigr)$,
where $\Delta = \tfrac{2}{n-1}$.
The incumbent is
$\widehat\policy_{N} = \argmin_{\ell \leq N} \rrrule(\policy_{\ell})$.

We refer to Algorithm~\ref{alg:anytime-rulebook} and both best-first B\&B variants collectively as
the \emph{box-based methods}.
For each optimizer, we record an anytime trace indexed by the iteration counter $N$.
Each snapshot contains the best incumbent $\widehat\policy_N$ and its excess risk
$\rrrule(\widehat\policy_N)$.
For the box-based methods, we additionally record the certified lower and upper certificates
$\underline\alpha(N)$ and $\overline\alpha(N)$, the certified optimality gap
$\varepsilon(N)=\overline\alpha(N)-\underline\alpha(N)$, the current box collection
$\mathcal O(N)$ together with the per-box bounds
$\underline q(B)$ and $\overline q(B)$, and the total Lebesgue measure of the cover,
$|\mathcal O(N)|=\sum_{B\in\mathcal O(N)}\prod_{j=1}^d l_j(B)$.
Since the uniform-grid baseline does not produce certified lower bounds, we use
$\varepsilon(N) = \rrrule(\widehat\policy_{N}) - \alpha(\policies)$,
in place of the certified optimality gap.

The computational budget is specified by the iteration limit $N_{\max}$.
For the box-based methods, one iteration performs one box bisection and two evaluations of $q$ at
the child-box centers, for a total of $2N_{\max}$ black-box evaluations.
For the uniform-grid baseline, one iteration evaluates $q$ at a single grid point.
Unless stated otherwise, we set $N_{\max}=900$ for all methods,
corresponding to a $30\times30$ grid.

\begin{figure}[t]
  \centering
  \includegraphics[width=0.48\columnwidth,trim={2mm 0 0 0},clip]{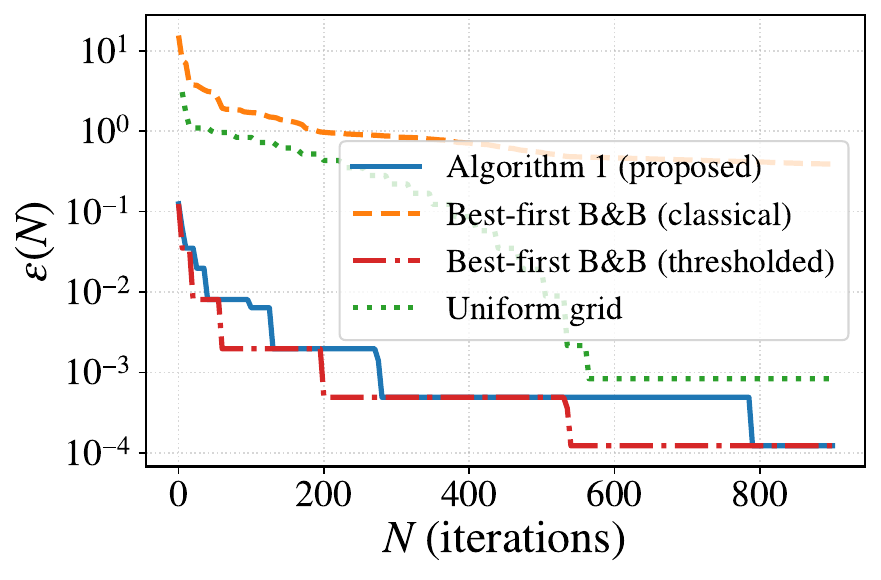}
  \includegraphics[width=0.48\columnwidth,trim={2mm 0 0 0},clip]{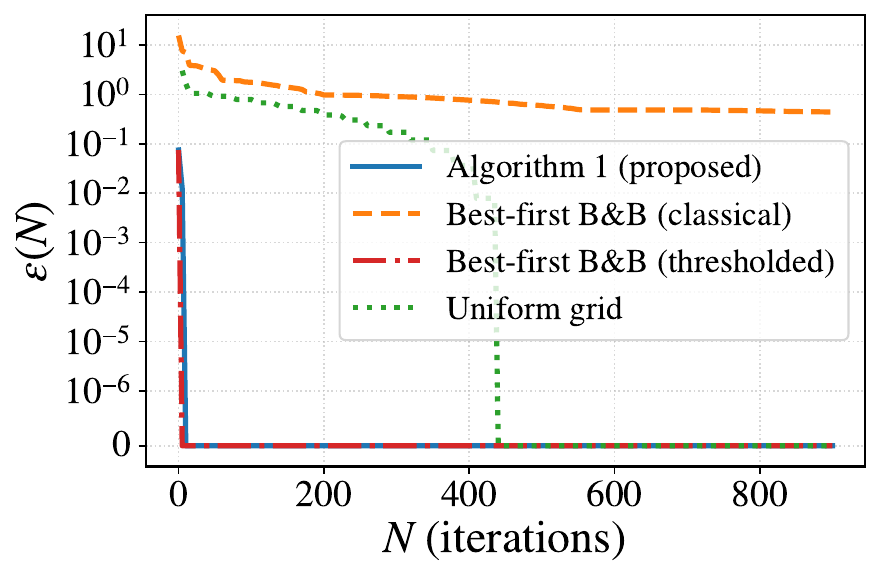}\\
  \includegraphics[width=0.48\columnwidth,trim={2mm 0 0 0},clip]{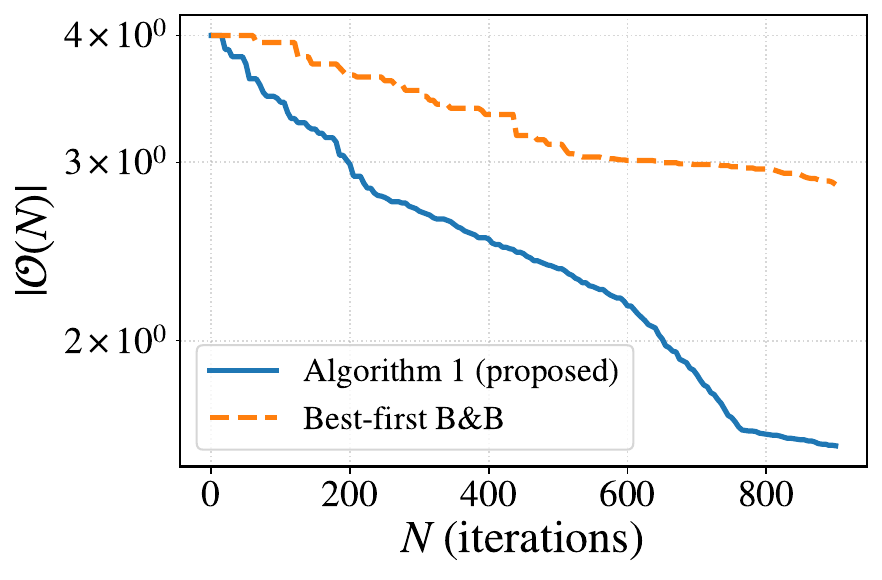}
  \includegraphics[width=0.48\columnwidth,trim={2mm 0 0 0},clip]{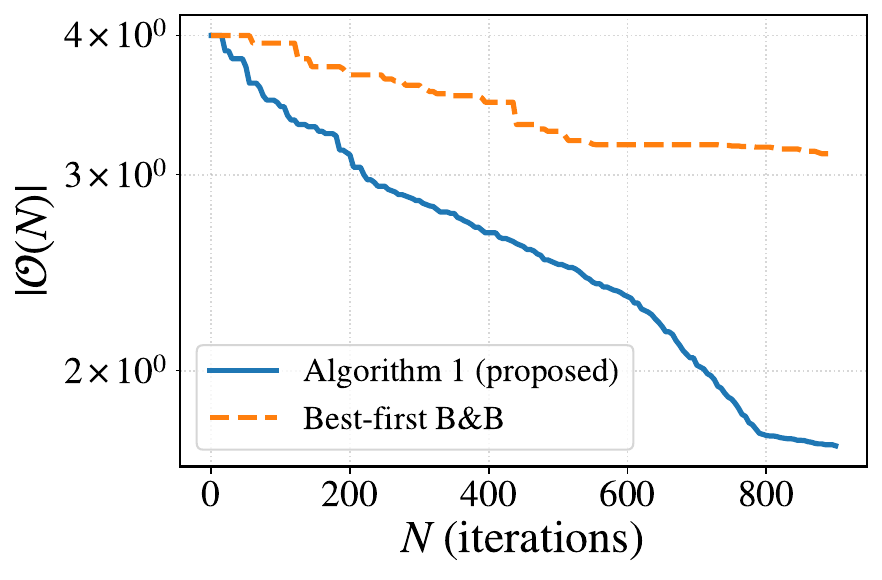}
  \caption{The shifted Rosenbrock-type single-rule scenario.
    The left and right columns correspond to $\threshold = 0$ and $\threshold = 0.05$.
    The top row shows the anytime gap $\varepsilon(N)$.
    The bottom row shows the cover measure $|\mathcal{O}(N)|$.
    The two best-first variants generate the same $\mathcal{O}$,
    so only one curve is shown.}
  \label{fig:risk-aware-gamma0}
\end{figure}

\noindent
\textbf{Results.}
Fig.~\ref{fig:risk-aware-gamma0} compares the 4 optimizers for
$\threshold = 0$ (satisfied only at $\policy^{\star}$)
and $\threshold = 0.05$ (satisfied on a nontrivial sublevel set).
The gap $\varepsilon$ of classical best-first B\&B remains large in both cases
because it repeatedly refines only the most promising box,
leaving many boxes with overly conservative lower bounds $\underline\rrrule(B)$.
Replacing this lower bound by the thresholded construction in~\eqref{eq:LBUB}
removes this conservatism without changing the box collection $\mathcal{O}$ or incumbent set $\mathcal{C}$
since the two variants perform the same refinements and generate the same incumbents.
Algorithm~\ref{alg:anytime-rulebook} achieves the same certified optimality gap as the thresholded baseline in
both experiments,
while the uniform-grid baseline reaches comparable incumbent quality but does not provide a
certified optimality gap.

The benefit of Algorithm~\ref{alg:anytime-rulebook} is most apparent in the cover measure
$|\mathcal O(N)|$, which it contracts substantially faster than both best-first B\&B variants.
This difference reflects the different objectives of the refinement strategies.
Best-first B\&B is designed to locate a global minimizer efficiently,
whereas Algorithm~\ref{alg:anytime-rulebook} is designed to certify
and eliminate regions that cannot contain lexicographically optimal policies.
The widest-refinable-box rule spreads refinement across the competitive region
so that regions far from $\policy^\star$ quickly become uncompetitive and can be pruned.
Best-first B\&B instead focuses refinement on the currently most promising box,
so pruning relies primarily on improvements in the incumbent value $\overline\alpha$.
When $\threshold = 0.05$, Algorithm~\ref{alg:anytime-rulebook} eventually switches from
\textsc{CPrune} to \textsc{TPrune} after an incumbent $\policy$ with $q(\policy) \leq \threshold$ is found,
but the algorithm continues to refine the widest ambiguous boxes and therefore maintains faster
cover contraction than best-first B\&B.

These experiments isolate the contributions of the proposed method.
Comparing the classical and thresholded best-first B\&B isolates the benefit of
exploiting the threshold structure,
while comparing the thresholded best-first B\&B with Algorithm~\ref{alg:anytime-rulebook} isolates
the benefit of the widest-refinable-box rule.
The former improves the certified optimality gap $\varepsilon$,
whereas the latter produces substantially tighter box collection $\mathcal O$.

\subsection{Highway Merging Simulation}
We next consider a highway merging scenario where an ego vehicle merges into an adjacent
lane between a front vehicle and a rear vehicle as shown in
Fig.~\ref{fig:bezier-parameterization}.
The complexity of the uncertain vehicle interactions makes the resulting risk-evaluation
functions impractical to characterize analytically,
so they are treated as black boxes.

\noindent
\textbf{Scenario.}
The ego vehicle starts in its current lane and must merge into a gap between front and rear
vehicles in the adjacent target lane over a fixed time horizon $[0,T]$.
We use $X$ and $Y$ to denote longitudinal and lateral coordinates in a straight-road frame aligned
with the $X$-axis.
The current- and target-lane centerlines are $Y=0$ and $Y=w_{\rm lane}$, respectively, and the
merging corridor is
$\mathcal L=[X_0,X_T]\times[-w_{\rm lane}/2,3w_{\rm lane}/2]$.

Let $\policy$ denote an ego policy, which determines the ego trajectory,
and let $\omega$ denote a realization of the surrounding vehicles' behavior,
which affects how they respond to the ego trajectory.
We define
collision avoidance rules $\rbrule_{1,k}$,
a rear-vehicle braking rule  $\rbrule_2$,
three-second headway rules $\rbrule_{3,k}$, and
comfort rules $\rbrule_{4,\rm acc}$ and $\rbrule_{4, \rm curv}$
for $k \in \{\rm front, \rm rear\}$ as
\begin{equation}
  \begin{aligned}
    \rbrule_{1,k}^\policy(\omega)
    &= \max_{t\in[0,T]} [d_{\rm col}-d_k^\policy(t,\omega)]_+,\\
    \rbrule_2^\policy(\omega)
    &= \max_{t\in[0,T]} [-a_{\rm rear}^\policy(t,\omega)]_+,\\
    \rbrule_{3,k}^\policy(\omega)
    &= \max_{t\in[0,T]}\chi^\policy(t)[T_{\rm gap}-h_{k}^\policy(t,\omega)]_+,\\
    \rbrule_{4,\rm acc}^\policy(\omega)
    &= \max_{t\in[0,T]} \left[\max\left\{
        \frac{a_{\rm ego}^\policy(t)}{a_{\rm comf}},
        \frac{-a_{\rm ego}^\policy(t)}{-b_{\rm comf}}\right\}-1
    \right]_+ \\
    \rbrule_{4,\rm curv}^\policy(\omega)
    &=
    \max_{t\in[0,T]}\left[\frac{|\kappa_{\rm ego}^\policy(t)|}{\kappa_{\rm comf}}-1\right]_+ .
  \end{aligned}
  \label{eq:ex-lane-change-rules}
\end{equation}
Here, $d_{k}^\policy(t,\omega)$ denote the signed clearance
from the ego vehicle to vehicle $k$.
Thus, $\rbrule_{1,k}$ measures violations of the required collision clearance $d_{\rm col}$.
The rear-vehicle acceleration $a_{\rm rear}^\policy(t,\omega)$ defines the braking
magnitude $[-a_{\rm rear}^\policy(t,\omega)]_+$, which is penalized by $\rbrule_2$.
The time headway $h_{k}^\policy$ quantifies the longitudinal gaps to vehicle $k$.
Thus, $\rbrule_{3}$ measures the violations of the desired minimum time headway $T_{\rm gap}$,
while the ego vehicle occupies the target lane, as determined by
the target-lane overlap factor $\chi^\policy(t)\in[0,1]$.
The comfort rules $\rbrule_{4,\rm acc}$ and $\rbrule_{4, \rm curv}$
depend only on the ego trajectory through its longitudinal acceleration
$a_{\rm ego}^\policy(t)$ and curvature $\kappa_{\rm ego}^\policy(t)$, respectively.
The constants $d_{\rm col}$, $T_{\rm gap}$, $a_{\rm comf}$, $b_{\rm comf}$, and
$\kappa_{\rm comf}$ specify the collision
clearance, desired time headway, comfortable acceleration,
braking, and curvature.
The numerical values of these constants are listed in
Table~\ref{tab:lane-change-parameters}.

The preorder relation is defined by
$r_{1,\rm front}, r_{1,\rm rear} > r_{2} > r_{3,\rm front}, r_{3,\rm rear} > r_{4,\rm acc}, r_{4,\rm curv}$,
where $\rbrule_{1,\rm front}$ and $\rbrule_{1,\rm rear}$ have equal rank,
as do $\rbrule_{3,\rm front}$ and $\rbrule_{3,\rm rear}$.
The two comfort rules $\rbrule_{4,\rm acc}$ and $\rbrule_{4,\rm curv}$ also have equal rank.
We refine this rulebook by aggregating rules with equal rank to obtain a total-order rulebook
$\rulebook_T$ with $\ruleset_T={\rbrule_1,\rbrule_2,\rbrule_3,\rbrule_4}$ and
$r_{1} > r_{2}> r_{3}> r_{4}$.
For $i\in\{1,3\}$, the front and rear rules are aggregated using the maximum:
$r_{i}^{\policy}(\omega) =
\max_{k\in\{\mathrm{front},\mathrm{rear}\}}\rbrule_{i,k}^\policy(\omega)$.
The two comfort rules are aggregated using a weighted sum:
$r_{4}^{\policy}(\omega) = \rbrule_{4,\rm acc}^{\policy}(\omega) + \lambda_{\rm curv}\rbrule_{4,\rm curv}^\policy(\omega)$,
where the dimensionless parameter $\lambda_{\rm curv}\geq0$ controls the relative weight
of curvature discomfort.

Under these definitions, insufficient surrounding-vehicle response may lead to violations  of
$\rbrule_1$ and $\rbrule_3$, whereas an overly aggressive rear-vehicle response may violate $\rbrule_2$.
These effects are captured through the corresponding risk-evaluation functions.
For the rules that involve uncertain interactions, we use CVaR as the risk measure:
$q_{i}(\policy) = \CVaR_{\beta_{i}}(\rbrule_{i}^{\policy}(\omega))$ for $i \in \{1, 2, 3\}$,
where $\beta_i\in(0,1)$ is the tail probability.
The comfort rule $\rbrule_4$ depends only on the ego trajectory, so
$\rbrule_4^\policy(\omega)=\rbrule_4^\policy$ for all $\omega$ and
$q_4(\policy)=\rbrule_4^\policy$.

\noindent
\textbf{Flatness-based policy parameterization.}
We exploit differential flatness to reduce the search dimension
and parameterize the flat output by a B\'{e}zier curve in
the Bernstein basis.
The kinematic bicycle model underlying this representation is
defined for a vehicle state $(X,Y,\psi,v)$, where $(X,Y)$ is the position of the rear-axle center, with steering
input $\delta$ and longitudinal acceleration input $a$:
\[
  \dot X = v\cos\psi, \hspace{2mm}
  \dot Y = v\sin\psi, \hspace{2mm}
  \dot\psi = \frac{v}{\ell_{\rm wb}}\tan\delta,
  \dot v = a .
\]
The flat output is the rear-axle position $z(t)=(X(t),Y(t))$.
Let $P_j=(x_j,y_j)\in\mathbb R^2$ denote the Bézier control points.
The degree-$n$ Bézier parameterization of the rear-axle trajectory is
$z_\theta(t)=\sum_{j=0}^{n}P_j(\theta)\beta_{j,n}(t/T)$,
where $\beta_{j,n}(\tau)=\binom{n}{j}\tau^j(1-\tau)^{n-j}$
is the Bernstein basis of degree $n$.
We choose $n=7$ so that, after fixing endpoint positions and velocities, four interior control
points remain as free search variables.
The endpoint position and velocity constraints are imposed by fixing
\begin{align*}
  P_0 &= (X_0,Y_0), &
  P_1 &= (X_0,Y_0) + \frac{T}{7}v_0 e(\psi_0),\\
  P_7 &= (X_T,Y_T), &
  P_6 &= (X_T,Y_T) - \frac{T}{7}v_T e(\psi_T),
\end{align*}
where $e(\psi)=(\cos\psi,\sin\psi)$.
This gives $\dot z_\theta(0)=v_0e(\psi_0)$ and $\dot z_\theta(T)=v_Te(\psi_T)$.
The free control points form the decision vector
$\theta=(P_{2}, \ldots, P_{5})\in\mathbb R^8$.

Let $\Delta X=X_T-X_0>0$, $Y_{\min}=-w_{\rm lane}/2$, and
$Y_{\max}=3w_{\rm lane}/2$.
For each free control point $j\in\{2,3,4,5\}$, define
\[
  x_j^- = X_0 + \frac{j-1/2}{7}\Delta X,
  \qquad
  x_j^+ = X_0 + \frac{j+1/2}{7}\Delta X .
\]
The box of free B\'{e}zier control-point coordinates is
\[
  \Theta_{\rm box}
  = \prod_{j=2}^5\bigl([x_j^-,x_j^+]\times[Y_{\min},Y_{\max}]\bigr).
\]
Thus, as illustrated in Fig.~\ref{fig:bezier-parameterization},
each free control point can move laterally across the full corridor width while remaining
within a longitudinal slab centered at its nominal progress location.
These slabs give each component of $\theta$ a clear geometric interpretation and encourage forward
progress without discretizing the full state-input trajectory.

\begin{figure}[t]
  \centering
  \begin{tikzpicture}[scale=0.78,>=Latex]
    \fill[gray!12] (0,-0.45) rectangle (10.4,1.45);
    \draw[gray!65] (0,-0.45) -- (10.4,-0.45);
    \draw[gray!65,dashed] (0,0.50) -- (10.4,0.50);
    \draw[gray!65] (0,1.45) -- (10.4,1.45);
    \node[gray!70,font=\scriptsize,align=center] at (9.00,0.05)
    {merging\\corridor $\mathcal L$};

    \foreach \x/\lab in {3.22/$P_2$,4.56/$P_3$,5.89/$P_4$,7.23/$P_5$} {
      \draw[orange!45,fill=orange!8] (\x-0.67,-0.28) rectangle (\x+0.67,1.28);
      \node[orange!70!black,font=\scriptsize] at (\x,1.57) {\lab};
    }

    \draw[fill=red!18,draw=red!60!black,rounded corners=0.5pt]
    (1.55,0.78) rectangle (2.25,1.10);
    \node[red!60!black,font=\scriptsize] at (1.90,0.60) {rear car};
    \draw[fill=green!18,draw=green!45!black,rounded corners=0.5pt]
    (8.95,0.78) rectangle (9.65,1.10);
    \node[green!45!black,font=\scriptsize] at (9.30,0.60) {front car};

    \coordinate (P0) at (0.55,0.05);
    \coordinate (P1) at (2.22,0.08);
    \coordinate (P2) at (3.22,0.17);
    \coordinate (P3) at (4.56,0.37);
    \coordinate (P4) at (5.89,0.71);
    \coordinate (P5) at (7.23,0.92);
    \coordinate (P6) at (8.23,0.95);
    \coordinate (P7) at (9.90,0.95);

    \draw[blue,thick] plot[smooth,tension=0.75]
    coordinates {(0.55,0.05) (2.25,0.08) (3.45,0.18) (4.85,0.44)
      (6.20,0.78) (7.55,0.92) (9.90,0.95)};

    \foreach \p in {P0,P1,P6,P7} {
      \fill[black] (\p) circle (1.5pt);
    }
    \foreach \p in {P2,P3,P4,P5} {
      \fill[orange!85!black] (\p) circle (1.7pt);
    }
    \node[above,font=\scriptsize] at (P0) {$P_0$};
    \node[below,font=\scriptsize] at (P1) {$P_1$};
    \node[above,font=\scriptsize] at (P6) {$P_6$};
    \node[above,font=\scriptsize] at (P7) {$P_7$};
    \node[blue,font=\scriptsize] at (5.35,0.35) {$z_\theta(t)$};
    \draw[<->,thin,gray!75] (2.25,1.82) -- (8.95,1.82);
    \node[gray!70,font=\scriptsize,fill=white,inner sep=1pt] at (5.60,2.02) {merge gap};
    \node[orange!70!black,font=\scriptsize,align=center] at (5.05,-0.72)
      {free parameters $\theta=(P_2,P_3,P_4,P_5)$};
  \end{tikzpicture}
  \caption{B\'{e}zier parameterization of the flat output for highway merging.
    The ego vehicle merges into the gap between the rear and front vehicles.
    The fixed control points $P_{0}$, $P_{1}$, $P_{6}$, $P_{7}$
    enforce the boundary positions and velocities.
    The free control points $P_2,\ldots,P_5$ are constrained to lie within the indicated
    longitudinal slabs and form the decision vector $\theta$.}
  \label{fig:bezier-parameterization}
\end{figure}
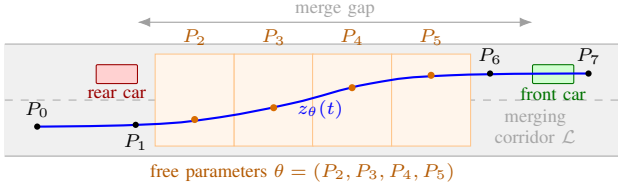

Assume that $\|\dot z_\theta\|_2>0$ over $[0,T]$.
Given any $\theta$, the vehicle state and inputs are recovered from $z_\theta$ by
\begin{align*}
  v_\theta &= \|\dot z_\theta\|_2, &
  \psi_\theta &= \operatorname{atan2}(\dot Y_\theta,\dot X_\theta),\\
  a_\theta &= \frac{\dot z_\theta^\top \ddot z_\theta}{\|\dot z_\theta\|_2}, &
  \delta_\theta &= \arctan(\ell_{\rm wb}\kappa_\theta),
\end{align*}
where the signed curvature of the rear-axle trajectory is
\begin{equation}
  \kappa_\theta =
  \frac{\dot X_\theta\ddot Y_\theta-\dot Y_\theta\ddot X_\theta}
       {\|\dot z_\theta\|_2^3} .
  \label{eq:bezier-curvature}
\end{equation}
The policy space for this experiment is
\[
  \begin{aligned}
    \policies=\{\theta\in\Theta_{\rm box}\mid
    & z_\theta(t)\in\mathcal L,
      \ v_\theta(t)\le v_{\max},\\
    & -b_{\max}\le a_\theta(t)\le a_{\max},\\
    & |\kappa_\theta(t)|\le\kappa_{\max},
      \ |\delta_\theta(t)|\le\delta_{\max},\\
    & \forall t\in[0,T]\}.
  \end{aligned}
\]
The constants $v_{\max}$, $a_{\max}$, $\kappa_{\max}$, and $\delta_{\max}$ are physical
limits, whereas $a_{\rm comf}\le a_{\max}$ and
$\kappa_{\rm comf}\le\kappa_{\max}$ define softer comfort targets in the rulebook.
The specific parameters used in the experiment are provided in
Table~\ref{tab:lane-change-parameters}.

\noindent
\textbf{Simulation-based oracle.}
In this experiment, the black-box oracle evaluates $q_i(\policy)$ using Monte
Carlo simulation of surrounding-vehicle interactions%
\footnote{In practice, this oracle is usually replaced by a neural network surrogate.}.
For each candidate policy $\policy\in\Theta_{\rm box}$ represented by
the decision vector $\theta$,
the oracle simply recovers the ego states and inputs from the
flatness-based parameterization:
$z^\policy=z_\theta$, $v^\policy=v_\theta$,
$\psi^\policy=\psi_\theta$, $a^\policy=a_\theta$,
$\kappa^\policy=\kappa_\theta$, and $\delta^\policy=\delta_\theta$.
It then checks the feasibility constraints defining $\policies$ and evaluates
the rules on the grid $t_\ell=\ell\Delta t$, $\ell=0,\ldots,K$, with
$K=T/\Delta t$.
At grid time $t_\ell$, the ego footprint center is
\begin{align*}
  c^\policy(t_\ell)
  &= (X_c^\policy(t_\ell),Y_c^\policy(t_\ell))\\
  &= z^\policy(t_\ell)
     +\ell_c(\cos\psi^\policy(t_\ell),\sin\psi^\policy(t_\ell)),
\end{align*}
where $\ell_c$ is the distance from the rear axle to the footprint center of the ego vehicle.
The ego quantities used in the rules \eqref{eq:ex-lane-change-rules} are
$a_{\rm ego}^\policy(t_\ell)=a^\policy(t_\ell)$,
$\kappa_{\rm ego}^\policy(t_\ell)=\kappa^\policy(t_\ell)$, and
$\chi^\policy(t_\ell) = \min\left\{1,
  \frac{[Y_c^\policy(t_\ell)-w_{\rm lane}/2+w_{\rm veh}/2]_+}{w_{\rm veh}}
\right\}$.

The oracle models each vehicle as a rectangle of length $\ell_{\rm veh}$
and width $w_{\rm veh}$.
For each Monte Carlo sample, a behavior realization $\omega$ is drawn
for each surrounding vehicle $k\in\{\mathrm{front},\mathrm{rear}\}$, specifying
its initial-state perturbations $\epsilon_{k}^S$, $\epsilon_{k}^v$,
response mode $M_{k} \in \{\rm{prompt}, \rm{delayed}, \rm{weak}, \rm{nonresponsive}\}$,
response delay $\tau_k$, standstill clearance $\Delta_{0,k}$,
desired time headway $T_k$, response gains $K_{g,k},K_{v,k}$,
maximum acceleration and braking magnitude $a_{k,\max},b_{k,\max}$,
and disturbance sequence $\xi_{k,\ell}$.
The initial longitudinal position and speed are sampled as
$S_{k,0}^\policy(\omega)
= \bar S_{k,0}+\epsilon_{k}^S$ and
$v_{k,0}^\policy(\omega)
= [\bar v_{k,0}+\epsilon_{k}^v]_+$,
where $\bar S_{k,0}$ and speed $\bar v_{k,0}$ are the nominal values listed in
Table~\ref{tab:lane-change-parameters}.
The response mode captures driver responsiveness:
$\rm{prompt}$ drivers have short reaction delay and nominal gains,
$\rm{delayed}$ drivers have longer delay,
$\rm{weak}$ drivers have reduced gains, and
$\rm{nonresponsive}$ drivers have zero gains.
Under $\omega$, vehicle $k$
evolves longitudinally in the target lane and, at time $t_{l}$,
is centered at $(S_{k,\ell}^\policy(\omega),w_{\rm lane})$ with zero heading,
speed $v_{k,\ell}^\policy(\omega)$, and acceleration $a_{k,\ell}^\policy(\omega)$.
At time $t_{l}$, vehicle $k$ perceives the ego state at the delayed time
$t_\ell^k=[t_\ell-\tau_k]_+$.
The perceived longitudinal gaps are
\begin{align*}
  \widehat g_{{\rm front},\ell}^\policy(\omega)
  &= S_{{\rm front},\ell}^\policy(\omega)-X_c^\policy(t_\ell^{\rm front})
     -\ell_{\rm veh},\\
  \widehat g_{{\rm rear},\ell}^\policy(\omega)
  &= X_c^\policy(t_\ell^{\rm rear})-S_{{\rm rear},\ell}^\policy(\omega)
     -\ell_{\rm veh}.
\end{align*}
Let $\Delta_{k,\ell}^{{\rm des},\policy}(\omega)
= \Delta_{0,k}+T_kv_{k,\ell}^\policy(\omega)$
denote the desired clearance of vehicle $k$,
where $\Delta_{0,k}$ is the standstill clearance and $T_kv_{k,\ell}^\policy$ is the
speed-dependent clearance induced by the desired time headway $T_{k}$.
The delayed ego speed perceived by vehicle $k$ is
$u_{k,\ell}^\policy=v^\policy(t_\ell^k)$.
Define the response terms
\begin{align*}
  \Gamma_{{\rm rear},\ell}^\policy(\omega)
  &= K_{g,{\rm rear}}
    (\widehat g_{{\rm rear},\ell}^\policy(\omega)
     -\Delta_{{\rm rear},\ell}^{{\rm des},\policy}(\omega))\\
  &\quad + K_{v,{\rm rear}}
    (u_{{\rm rear},\ell}^\policy-v_{{\rm rear},\ell}^\policy(\omega)),\\
  \Gamma_{{\rm front},\ell}^\policy(\omega)
  &= K_{g,{\rm front}}
    [\Delta_{{\rm front},\ell}^{{\rm des},\policy}(\omega)
     -\widehat g_{{\rm front},\ell}^\policy(\omega)]_+\\
  &\quad + K_{v,{\rm front}}
    [u_{{\rm front},\ell}^\policy-v_{{\rm front},\ell}^\policy(\omega)]_+ .
\end{align*}
The gains $K_{g,k}$ and $K_{v,k}$ characterize driver responsiveness to the desired clearance and to
the speed of the ego vehicle, respectively.
To enforce the acceleration and braking limits,
define $\operatorname{clip}_k(x)=\min\{a_{k,\max},\max\{-b_{k,\max},x\}\}$.
The acceleration of vehicle $k$ is given by
$a_{k,\ell}^\policy(\omega)
= \operatorname{clip}_k\bigl(
\xi_{k,\ell}
+\chi_{k,\ell}^\policy\Gamma_{k,\ell}^\policy(\omega)\bigr)$, where
the disturbance $\xi_{k,\ell}$ represents unmodeled driver variation
and the delayed target-lane overlap factor
$\chi_{k,\ell}^\policy=\chi^\policy(t_\ell^k)$ activates the response as the ego enters the target
lane.
The longitudinal states are propagated using the standard kinematic model:
for $\ell=0,\ldots,K-1$,
\begin{align*}
  S_{k,\ell+1}^\policy(\omega)
  &= S_{k,\ell}^\policy(\omega)+v_{k,\ell}^\policy(\omega)\Delta t
    +\tfrac12 a_{k,\ell}^\policy(\omega)\Delta t^2,\\
  v_{k,\ell+1}^\policy(\omega)
  &= \max\{0,v_{k,\ell}^\policy(\omega)+a_{k,\ell}^\policy(\omega)\Delta t\},
\end{align*}

We now define the quantities used in the rule definitions \eqref{eq:ex-lane-change-rules}.
Let $\mathcal{R}_{\rm{ego},\ell}^{\policy}$ denote a rectangle of length $\ell_{\rm veh}$ and width
$w_{\rm veh}$, centered at $c^\policy(t_\ell)$ with heading $\psi^\policy(t_\ell)$,
and let $\mathcal{R}_{k,\ell}^{\policy}(\omega)$ denote an axis-aligned rectangle,
centered at $(S_{k,\ell}^\policy(\omega),w_{\rm lane})$.
The signed clearances $d_k^\policy(t_\ell,\omega)$ in $\rbrule_{1}$
is the signed distance between
$\mathcal{R}_{\rm{ego},\ell}^{\policy}$ and $\mathcal{R}_{k,\ell}^{\policy}(\omega)$,
positive when they are disjoint and negative when they overlap.
The headways used in $\rbrule_3$ are defined as
\[
  h_{\rm front}^\policy(t_\ell,\omega)
  = \frac{g_{\rm front}^\policy(t_\ell,\omega)}
  {v^\policy(t_\ell)}, \hspace{3mm}
  h_{\rm rear}^\policy(t_\ell,\omega)
  = \frac{g_{\rm rear}^\policy(t_\ell,\omega)}
  {v_{{\rm rear},\ell}^\policy(\omega)},
\]
with the convention that the headway is simply set to $T_{\rm gap}$
when the corresponding speed is zero.
Here,
\begin{align*}
  g_{\rm front}^\policy(t_\ell,\omega)
  &= S_{{\rm front},\ell}^\policy(\omega)-X_c^\policy(t_\ell)-\ell_{\rm veh},\\
  g_{\rm rear}^\policy(t_\ell,\omega)
  &= X_c^\policy(t_\ell)-S_{{\rm rear},\ell}^\policy(\omega)-\ell_{\rm veh}.
\end{align*}
are the realized longitudinal gaps.
A risky ego merge can therefore produce a collision or headway violation when a sampled driver has
large delay, weak gains, unfavorable disturbances, or no response.

Before optimization begins, we draw a fixed scenario set
$\Omega_{\rm MC}=\{\omega_1,\ldots,\omega_{M_{\rm MC}}\}$ and reuse it for every policy evaluation.
This makes $q_{i}$ a deterministic black-box function.
For a fixed policy $\policy$, the oracle rolls out each scenario $\omega_{j} \in \Omega_{\rm MC}$,
computes $\rbrule_i^\policy(\omega_j)$ for $i\in\{1,2,3\}$ based on \eqref{eq:ex-lane-change-rules},
and returns the empirical CVaR at tail probability $\beta_{i}$:
\[
  q_i(\policy)=\min_{\eta\in\mathbb R}\left\{
    \eta + \frac{1}{\beta_{i} M_{\rm MC}}\sum_{j=1}^{M_{\rm MC}}
    [\rbrule_i^\policy(\omega_j)-\eta]_+\right\}.
\]
The comfort rule is deterministic, so $q_4(\policy)=\rbrule_4^\policy$ is evaluated directly from
$z^\policy$.
Since the computation of $q_{i}$ is external to Algorithm~\ref{alg:anytime-rulebook},
we use the number of $q_{i}$ evaluations as the primary black-box budget.
This separates oracle usage from the cost of the optimization procedure.
For computation time, we exclude the time spent evaluating $q_{i}$
and report only the time spent executing Algorithm~\ref{alg:anytime-rulebook}.


\noindent
\textbf{Lipschitz constants.}
Following \cite{Wood:1996:LipschitzEstimation,Malherbe:2017:LIPO},
we employ a sampling-based method to estimate $L_{i}$ for each $i \in \{1, \ldots, 4\}$.
Note that each $\rbrule_{i}^{\policy}(\omega_{j})$ is Lipschitz in $\policy$ over $\Theta_{\rm box}$
because the flatness-recovery and rule computations are Lipschitz on
$\Theta_{\rm box}$, provided the relevant denominators are bounded away from zero.
We therefore draw a set of i.i.d. policy pairs
$\{(\policy_{a}^{(k)},\policy_{b}^{(k)})\}_{k=1}^{N_{\rm pair}}$.
For each $\rbrule_{i} \in \ruleset$ and $\omega_{j} \in \Omega_{\rm MC}$,
we compute the maximum pairwise slope
$\widehat L_{ij}
= \max_{k}
\frac{\bigl|\rbrule_{i}^{\policy_{a}^{(k)}}(\omega_{j})-\rbrule_{i}^{\policy_{b}^{(k)}}(\omega_{j})\bigr|}
{\|\policy_{a}^{(k)}-\policy_{b}^{(k)}\|_{2}}$,
which provides a lower bound on the corresponding Lipschitz constant $L_{ij}$.
Since empirical CVaR cannot change by more than the largest change in any
$\rbrule_{i}^{\policy}(\omega_{j})$,
we can conclude that $q_{i}$ is Lipschitz with constant $L_{i} \leq \max_{j} L_{ij}$.
In practice, we use the inflated estimate
$L_{i} = \alpha \max_{j} \widehat L_{ij}$,
for some $\alpha > 1$.



\begin{table*}[t]
  \caption{Numerical values used in the highway merging simulation.}
  \label{tab:lane-change-parameters}
  \centering
  \scriptsize
  \setlength{\tabcolsep}{4pt}
  \begin{tabular}{p{0.12\textwidth}p{0.22\textwidth}p{0.6\textwidth}}
    \hline
    Group & Quantity & Value \\
    \hline
    Scenario and vehicle
          & Horizon, grid and lane width
                     & $T=5{\rm s}$, $\Delta t=0.1{\rm s}$, $w_{\rm lane}=3.6{\rm m}$ \\
          & Ego boundary conditions
                     & $(X_0,Y_0)=(0,0)$, $(X_T,Y_T)=(65,w_{\rm lane})$, $\psi_0=\psi_T=0$, $v_0=v_T=15{\rm m/s}$ \\
          & Vehicle dimensions
                     & $\ell_{\rm veh}=4.5{\rm m}$, $w_{\rm veh}=2.0{\rm m}$, $\ell_{\rm wb}=2.8{\rm m}$, $\ell_c=1.25{\rm m}$ \\
    \hline
    Physical limits
          & Speed, acceleration, curvature, steering
                     & $v_{\max}=25{\rm m/s}$, $a_{\max}=6{\rm m/s^2}$, $b_{\max}=8{\rm m/s^{2}}$ $\kappa_{\max}=0.25{\rm m^{-1}}$, $\delta_{\max}=0.5{\rm rad}$ \\
    \hline
    Rulebook constants
          & Collision, headway, and comfort
                     & $d_{\rm col}=0$, $T_{\rm gap}=3{\rm s}$,
                       $a_{\rm comf}=1.5{\rm m/s^2}$, $b_{\rm comf}=3.0{\rm m/s^2}$,
                       $\kappa_{\rm comf}=0.15{\rm m^{-1}}$, $\lambda_{\rm curv}=1$ \\
          & Rule thresholds
                     & $\threshold_1=0.015{\rm m}$, $\threshold_2=2.75{\rm m/s^2}$, $\threshold_3=2.0{\rm
                       s}$, $\threshold_4=0$ \\
    \hline
    Surrounding vehicles
          & Nominal initial positions and speeds
                     & $\bar S_{{\rm front},0}=45{\rm m}$, $\bar S_{{\rm rear},0}=-35{\rm m}$, $\bar
                       v_{{\rm front},0}=\bar v_{{\rm rear},0}=15{\rm m/s}$, which is $\sim 5.3{\rm
                       s}$ time gap \\
          & Initial perturbations
                     & $\epsilon_k^S\sim\mathcal N(0,1^2){\rm m}$, $\epsilon_k^v\sim\mathcal
                       N(0,0.5^2){\rm m/s}$ \\
          & Physical limits
                     & $b_{\rm rear, \max}=6{\rm m/s^2}$,
                       $b_{\rm front, \max}=3{\rm m/s^2}$,
                       $a_{\rm rear, \max}=2{\rm m/s^2}$,
                       $a_{\rm front, \max}=1{\rm m/s^2}$\\
          & Disturbances
                     & sample $\tilde\xi_{k,\ell}\sim\mathcal N(0,0.5^2){\rm m/s^2}$
                       and clip to $[-1.5,1.0]{\rm m/s^2}$\\
    \hline
    Rear Response
          & Desired clearance
                     & $\Delta_{0,\rm rear}=5{\rm m}$,
                       $T_{\rm rear}\sim\operatorname{Unif}[1.0,1.6]{\rm s}$ \\
          & Mode probabilities
                     &$\Pr(M_{\rm rear}=\mathrm{prompt})=0.45$,
                       $\Pr(M_{\rm rear}=\mathrm{delayed})=0.25$,\\
          &&$\Pr(M_k=\mathrm{weak})=0.20$,
             $\Pr(M_k=\mathrm{nonresponsive})=0.10$\\
          & Prompt
                     & $\tau_{\rm rear}\sim\operatorname{Unif}[0.2,0.5]{\rm s}$,
                       $K_{g,{\rm rear}}\sim\operatorname{Unif}[0.05,0.15]{\rm s^{-2}}$,
                       $K_{v,{\rm rear}}\sim\operatorname{Unif}[0.2,0.4]{\rm s^{-1}}$ \\
          & Delayed
                     &$\tau_{\rm rear}\sim\operatorname{Unif}[0.8,1.5]{\rm s}$,
                       $K_{g,{\rm rear}}\sim\operatorname{Unif}[0.05,0.15]{\rm s^{-2}}$,
                       $K_{v,{\rm rear}}\sim\operatorname{Unif}[0.2,0.4]{\rm s^{-1}}$ \\

          & Weak
                     & $\tau_{\rm rear}\sim\operatorname{Unif}[0.2,0.5]{\rm s}$
                       $K_{g,{\rm rear}}\sim\operatorname{Unif}[0.0125,0.0375]{\rm s^{-2}}$,
                       $K_{v,{\rm rear}}\sim\operatorname{Unif}[0.05,0.1]{\rm s^{-1}}$ \\
          & Nonresponsive
                     & $\tau_{\rm rear}\sim\operatorname{Unif}[0.2,0.5]{\rm s}$,
                       $K_{g,{\rm rear}}=K_{v,{\rm rear}}=0$\\
    \hline
    Front Response
          & Desired clearance
                     & $\Delta_{0,\rm front}=5{\rm m}$,
                       $T_{\rm front}\sim\operatorname{Unif}[1.0,1.6]{\rm s}$ \\
          & Mode probabilities
                     &$\Pr(M_{\rm front}=\mathrm{prompt})=0.30$,
                       $\Pr(M_{\rm front}=\mathrm{delayed})=0.20$,\\
          &&$\Pr(M_{\rm front}=\mathrm{weak})=0.20$,
             $\Pr(M_{\rm front}=\mathrm{nonresponsive})=0.30$\\
          & Prompt
                     & $\tau_{\rm front}\sim\operatorname{Unif}[0.2,0.5]{\rm s}$,
                       $K_{g,{\rm front}}\sim\operatorname{Unif}[0.0125,0.0375]{\rm s^{-2}}$,
                       $K_{v,{\rm front}}\sim\operatorname{Unif}[0.05,0.1]{\rm s^{-1}}$ \\
          & Delayed
                     &$\tau_{\rm front}\sim\operatorname{Unif}[0.8,1.5]{\rm s}$,
                       $K_{g,{\rm front}}\sim\operatorname{Unif}[0.0125,0.0375]{\rm s^{-2}}$,
                       $K_{v,{\rm front}}\sim\operatorname{Unif}[0.05,0.1]{\rm s^{-1}}$ \\

          & Weak
                     & $\tau_{\rm front}\sim\operatorname{Unif}[0.2,0.5]{\rm s}$
                       $K_{g,{\rm front}}\sim\operatorname{Unif}[0.003125,0.009375]{\rm s^{-2}}$,\\
          &&$K_{v,{\rm front}}\sim\operatorname{Unif}[0.0125,0.025]{\rm s^{-1}}$ \\
          & Nonresponsive
                     & $\tau_{\rm front}\sim\operatorname{Unif}[0.2,0.5]{\rm s}$,
                       $K_{g,{\rm front}}=K_{v,{\rm front}}=0$\\
    \hline
    Monte Carlo oracle
    & Empirical CVaR
                     &$\beta_{1} = 0.05$, $\beta_{2}=0.1$, $\beta_{3}=0.1$, $M_{\rm MC}=200$ \\
    \hline
  \end{tabular}
\end{table*}

\noindent
\textbf{Implementation of Algorithm~\ref{alg:anytime-rulebook}.}
We initialize $\mathcal{O} = \{\Theta_{\rm box}\}$
and $\mathcal{C} = \{\policy_{0}\}$, where
$\policy_{0} \in \policies$ is chosen as a smooth trajectory
following standard techniques for B\'{e}zier curves~\cite{Farin:2001:Curves},
without considering the surrounding vehicles.
The initial excess risk is $\rrrule(\policy_{0}) = (0, 0.28, 0, 0.6)$.
The algorithm is assigned an overall computational budget of $50{\rm ms}$ across all rule levels,
corresponding to a $20{\rm Hz}$ planning rate for real-time operation.
The cumulative computation time includes all operations of Algorithm~\ref{alg:anytime-rulebook},
excluding the time spent evaluating $q_{i}$ through the black-box oracle.
The refinement loop terminates when either the overall time budget is reached or,
for rule levels before the final level, $N_{\rm iter}=50$ iterations are completed.
Once the time budget is reached, the algorithm completes the remaining operations for the current
level and proceeds through subsequent levels, where the refinement loop terminates immediately
without further refinement.
These non-refinement operations add only a small computational overhead.
We use the center-based construction in~\eqref{eq:LBUB} to compute the bounds in
\textsc{Bounds}.
The \textsc{CSamples} helper uses the center of each given box as a candidate and
returns those certified to lie in $D_{i-1}$.

\noindent
\textbf{Results.}
Algorithm~\ref{alg:anytime-rulebook} takes an average of $50.3{\rm ms}$
and issues an average of $370$ oracle calls in total,
with a variation of $\pm 5\%$ across runs.
The returned policy is identical across runs as the final policy is reached within the first
$22{\rm ms}$.

The collision-avoidance level $r_{1}$ starts with
$\overline\alpha_{1} = 0$, as initialized by \textsc{InitCert},
and terminates its refinement loop
in the first iteration at approximately $0.01{\rm ms}$
because \textsc{TRefine} returns an empty refinable set $\mathcal R=\emptyset$.
No refinement is performed, so $\mathcal C = \{\policy_{0}\}$
remains unchanged at the end of the level.
The rear-braking level $r_{2}$ then starts with
$\overline\alpha_2=0.28{\rm m/s^{2}}$
and therefore proceeds in the \textsc{CPrune} branch.
As shown in Figure~\ref{fig:lane-change-convergence-gamma3-beta3},
$\overline\alpha_{2}$ decreases to
$0.06$, and $0{\rm m/s^{2}}$
at $0.46$ and $0.66{\rm ms}$, respectively
(iterations $11$ and $15$).
At the first occurrence of $\overline\alpha_{2}=0$,
the corresponding $q_2$ value is $0.07{\rm m/s^{2}}$ below $\threshold_2$
and the refinement loop switches to the \textsc{TPrune} branch.
Although some subsequent candidates constructed during refinement have smaller values of $q_2$,
$\overline\alpha_2$ remains zero due to the threshold structure.
Throughout the level, the lower certificate remains $\underline\alpha_{2} = 0$ and
the incumbent set grows to $|\mathcal{C}| = 37$.
Level $r_{2}$ terminates at approximately $3.8{\rm ms}$ upon reaching
the $N_{\rm iter}$ iteration limit.
At the beginning of the headway level $r_{3}$,
\textsc{InitCert} finds a policy $\policy \in \mathcal C$ with $q_{3}(\policy) \leq \threshold_{3}$,
and thus initializes $\overline\alpha_{3} = 0$.
The refinement loop then executes the \textsc{TPrune} branch at each iteration
and terminates at approximately $15.6{\rm ms}$ with $|\mathcal{C}| = 10$ upon reaching
the $N_{\rm iter}$ iteration limit.
Finally, the comfort level $r_{4}$ is initialized
with $\overline\alpha_{4} = 0.46$.
The level then runs for the remaining $\sim34{\rm ms}$ with
$\overline\alpha_{4}$ decreases to $0.45$, $0.44$, and $0.40$ at
$21.2$, $21.5$, and $21.9{\rm ms}$, respectively
(iterations $16$, $17$, and $18$).
The final incumbent $\policy^{\star}=\textsc{Best}{(\mathcal C,4)}$
has $\rrrule(\policy^{\star}) = (0, 0, 0, 0.40)$, i.e.,
rules $r_{1}$, $r_{2}$ and $r_{3}$ are all satisfied and $r_{4}$ is violated by $0.40$.
Figs.~\ref{fig:lane-change-convergence-gamma3-beta3} shows the
upper certificates $\overline\alpha_i$ and the size $|\mathcal{C}|$
of the incumbent set as computation time proceeds.

\begin{figure}[t]
  \centering
  \includegraphics[width=\columnwidth]{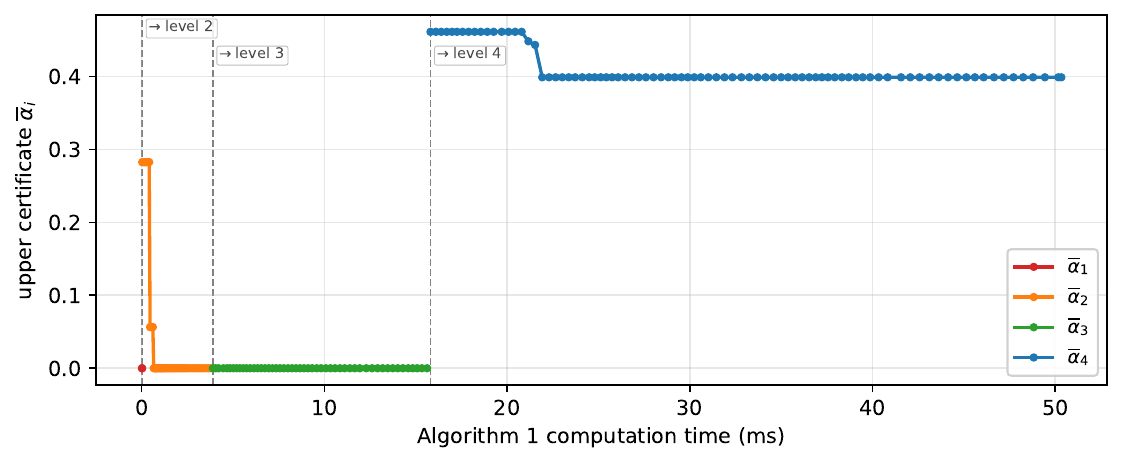}
  \includegraphics[width=\columnwidth]{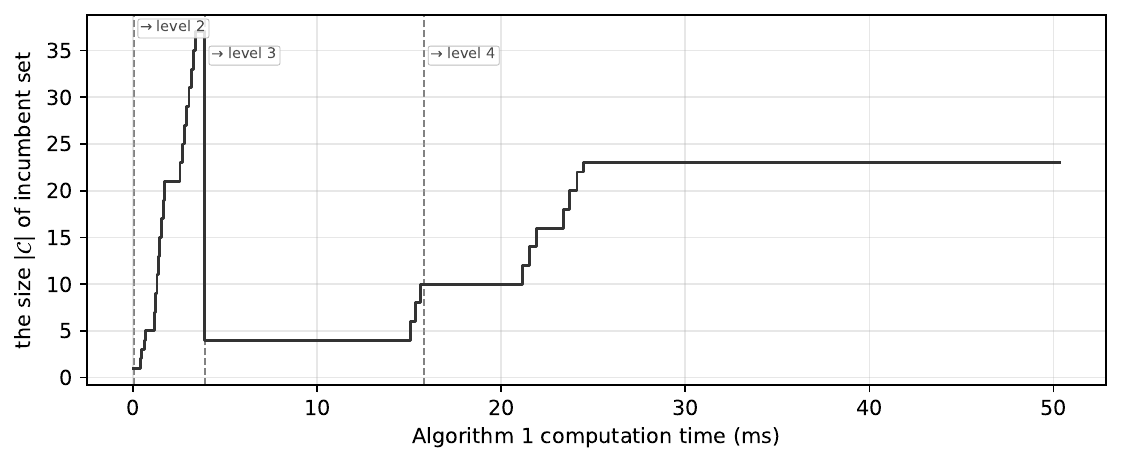}
  \caption{
    Upper certificate $\overline\alpha_i$ for $i=1,\dots,4$ (top) and
    the size $|\mathcal{C}|$ of the incumbent set (bottom)
    vs. Algorithm~\ref{alg:anytime-rulebook}'s computation time.
    $\overline\alpha_{i}$ is recorded at each snapshot while
    level $r_{i}$ is processed.
    $|\mathcal{C}|$ increases during the refinement of each rule level
    and decreases sharply at the end of the level due to the filtering
    $\mathcal{C} \cap \mathcal{D}_{i}$ on line \ref{alg:filter-incumbents}.
  }
  \label{fig:lane-change-convergence-gamma3-beta3}
\end{figure}

\section{Conclusions and Future Work}
This paper presented an anytime algorithm for risk-aware optimal control with rulebooks when the
risk-evaluation functions are treated as black boxes.
The formulation compares policies lexicographically using thresholded excess risks,
allowing each rule to have its own risk measure and threshold within the rulebook priority order.
The algorithm combines Lipschitz branch-and-bound with filtering.
The theoretical results show that the certificates returned by the algorithm provide valid bounds
for any finite computational budget.
Under additional assumptions, the certified optimality gap converges to zero
and every limit point of the returned incumbents is optimal.
The synthetic benchmark illustrates how exploiting the threshold structure and
widest-refinable-box selection improve certificate quality and cover contraction,
while the highway-merging simulation demonstrates
the method in a realistic application.

Several directions remain open for future work.
First, the convergence analysis relies on conservative global Lipschitz constants and sampling
assumptions.
Developing weaker convergence conditions and adaptive local Lipschitz estimates could reduce
conservatism.
Approaches that avoid requiring a known Lipschitz constant altogether provide another promising
direction \cite{Jones:1993:Lipschitzian}.
Second, the current certificates assume deterministic black-box evaluations of $q_i$.
When $q_i$ is estimated from finite Monte Carlo samples or neural surrogates,
statistical uncertainty and surrogate approximation error needs to be incorporated into the box
bounds.
Finally, learning different elements of rulebooks (e.g., rules, thresholds, priorities)
from expert demonstrations or observed safety outcomes will
enable the framework to adapt to different environments and operating conditions.

\bibliographystyle{ieeetr}
\bibliography{ref}

@inproceedings{Censi:2019:Liability,
  author={Andrea Censi and Konstantin Slutsky and Tichakorn Wongpiromsarn and  Dmitry Yershov and Scott Pendleton and James Fu and Emilio Frazzoli},
  booktitle={2019 International Conference on Robotics and Automation (ICRA)},
  title={Liability, Ethics, and Culture-Aware Behavior Specification using Rulebooks},
  year={2019},
  volume={},
  number={},
  pages={8536-8542}
}

@ARTICLE{Wongpiromsarn:2026:Formal,
  author={Wongpiromsarn, Tichakorn and Slutsky, Konstantin and Frazzoli, Emilio},
  journal={IEEE Transactions on Robotics},
  title={Formal Specification and Control Synthesis of Autonomous Robots Using Rulebooks},
  year={2026},
  volume={42},
  number={},
  pages={1330-1350},
  doi={10.1109/TRO.2026.3663976}}

@book{Bertsekas:2019:Reinforcement,
  title     = {Reinforcement Learning and Optimal Control},
  author    = {Bertsekas, Dimitri P.},
  year      = {2019},
  publisher = {Athena Scientific},
  isbn      = {978-1-886529-39-7}
}

@book{Sutton:2018:Reinforcement,
  author    = {Richard S. Sutton and Andrew G. Barto},
  title     = {Reinforcement Learning: An Introduction},
  edition   = {2nd},
  publisher = {MIT Press},
  year      = {2018},
}

@article{Artzner:1999:Coherent,
author = {Artzner, Philippe and Delbaen, Freddy and Eber, Jean-Marc and Heath, David},
title = {Coherent Measures of Risk},
journal = {Mathematical Finance},
volume = {9},
number = {3},
pages = {203-228},
doi = {https://doi.org/10.1111/1467-9965.00068},
year = {1999}
}

@article{Rockafellar:2000:Optimization,
  author={Rockafellar, R. Tyrrell and Uryasev, Stanislav},
  title={Optimization of Conditional Value-at-Risk},
  journal={The Journal of Risk},
  volume={2},
  number={3},
  pages={21--41},
  year={2000},
  doi={10.21314/JOR.2000.038}
}

@InCollection{Pflug:2000:Some,
author="Pflug, Georg Ch.",
editor="Uryasev, Stanislav P.",
title="Some Remarks on the Value-at-Risk and the Conditional Value-at-Risk",
bookTitle="Probabilistic Constrained Optimization: Methodology and Applications",
year="2000",
publisher="Springer US",
address="Boston, MA",
pages="272--281",
isbn="978-1-4757-3150-7",
doi="10.1007/978-1-4757-3150-7_15",
url="https://doi.org/10.1007/978-1-4757-3150-7_15"
}

@InProceedings{Majumdar:2020:How,
author="Majumdar, Anirudha
and Pavone, Marco",
editor="Amato, Nancy M.
and Hager, Greg
and Thomas, Shawna
and Torres-Torriti, Miguel",
title="How Should a Robot Assess Risk? Towards an Axiomatic Theory of Risk in Robotics",
booktitle="Robotics Research",
year="2020",
publisher="Springer International Publishing",
address="Cham",
pages="75--84",
isbn="978-3-030-28619-4"
}

@article{Lindemann:2023:Risk,
author = {Lindemann, Lars and Jiang, Lejun and Matni, Nikolai and Pappas, George J.},
title = {Risk of Stochastic Systems for Temporal Logic Specifications},
year = {2023},
issue_date = {May 2023},
publisher = {Association for Computing Machinery},
address = {New York, NY, USA},
volume = {22},
number = {3},
issn = {1539-9087},
url = {https://doi.org/10.1145/3580490},
doi = {10.1145/3580490},
journal = {ACM Trans. Embed. Comput. Syst.},
month = apr,
articleno = {54},
numpages = {31}
}

@ARTICLE{Wongpiromsarn:2026:Risk,
  author={Wongpiromsarn, Tichakorn},
  journal={IEEE Control Systems Letters},
  title={Risk-Aware Rulebooks for Multi-Objective Trajectory Evaluation Under Uncertainty},
  year={2026},
  volume={10},
  number={},
  pages={397-402},
  doi={10.1109/LCSYS.2026.3696351}}

@inproceedings{Turchetta:2019:Safe,
 author = {Turchetta, Matteo and Berkenkamp, Felix and Krause, Andreas},
 booktitle = {Advances in Neural Information Processing Systems},
 editor = {H. Wallach and H. Larochelle and A. Beygelzimer and F. d\textquotesingle Alch\'{e}-Buc and E. Fox and R. Garnett},
 pages = {},
 publisher = {Curran Associates, Inc.},
 title = {Safe Exploration for Interactive Machine Learning},
 url = {https://proceedings.neurips.cc/paper_files/paper/2019/file/4f398cb9d6bc79ae567298335b51ba8a-Paper.pdf},
 volume = {32},
 year = {2019}
}

@inproceedings{Chow:2015:RiskSensitive,
 author = {Chow, Yinlam and Tamar, Aviv and Mannor, Shie and Pavone, Marco},
 booktitle = {Advances in Neural Information Processing Systems},
 editor = {C. Cortes and N. Lawrence and D. Lee and M. Sugiyama and R. Garnett},
 pages = {},
 publisher = {Curran Associates, Inc.},
 title = {Risk-Sensitive and Robust Decision-Making: a {CVaR} Optimization Approach},
 url = {https://proceedings.neurips.cc/paper_files/paper/2015/file/64223ccf70bbb65a3a4aceac37e21016-Paper.pdf},
 volume = {28},
 year = {2015}
}

@article{Chow:2018:RiskConstrained,
  author  = {Yinlam Chow and Mohammad Ghavamzadeh and Lucas Janson and Marco Pavone},
  title   = {Risk-Constrained Reinforcement Learning with Percentile Risk Criteria},
  journal = {Journal of Machine Learning Research},
  year    = {2018},
  volume  = {18},
  number  = {167},
  pages   = {1-51},
  url     = {http://jmlr.org/papers/v18/15-636.html}
}

@inproceedings{Ahmadi:2021:RiskAverse,
author={Ahmadi, Mohamadreza and Dixit, Anushri and Burdick, Joel W. and Ames, Aaron D.},
booktitle={2021 60th IEEE Conference on Decision and Control (CDC)},
title={Risk-Averse Stochastic Shortest Path Planning},
year={2021},
volume={},
number={},
pages={5199-5204},
doi={10.1109/CDC45484.2021.9683527}
}

@inproceedings{Rigter:2022:Planning,
  title={Planning for Risk-Aversion and Expected Value in {MDPs}},
  author={Rigter, Marc and Duckworth, Paul and Lacerda, Bruno and Hawes, Nick},
  booktitle={32nd International Conference on Automated Planning and Scheduling (ICAPS)},
  volume={32},
  pages={307--315},
  year={2022},
  doi={10.1609/icaps.v32i1.19814}
}

@inproceedings{Wray:2015:Multi-Objective,
author = {Wray, Kyle Hollins and Zilberstein, Shlomo and Mouaddib, Abdel-Illah},
title = {Multi-objective {MDPs} with conditional lexicographic reward preferences},
year = {2015},
isbn = {0262511290},
publisher = {AAAI Press},
booktitle = {Proceedings of the Twenty-Ninth AAAI Conference on Artificial Intelligence},
pages = {3418-3424},
numpages = {7},
location = {Austin, Texas},
series = {AAAI'15},
doi = {10.1609/aaai.v29i1.9647}
}

@article{Rasekhipour:2018:Lexicographic,
title = {Autonomous driving motion planning with obstacles prioritization using lexicographic optimization},
journal = {Control Engineering Practice},
volume = {77},
pages = {235-246},
year = {2018},
issn = {0967-0661},
doi = {https://doi.org/10.1016/j.conengprac.2018.04.014},
url = {https://www.sciencedirect.com/science/article/pii/S0967066118300959},
author = {Yadollah Rasekhipour and Iman Fadakar and Amir Khajepour}
}

@article{Hakobyan:2019:Risk,
  author={Hakobyan, Astghik and Kim, Gyeong Chan and Yang, Insoon},
  journal={IEEE Robotics and Automation Letters},
  title={Risk-Aware Motion Planning and Control Using {CVaR}-Constrained Optimization},
  year={2019},
  volume={4},
  number={4},
  pages={3924-3931},
  doi={10.1109/LRA.2019.2929980}
}

@inproceedings{Nyberg:2021:Risk,
  author={Nyberg, Truls and Pek, Christian and Dal Col, Laura and Norén, Christoffer and Tumova, Jana},
  booktitle={2021 IEEE Intelligent Vehicles Symposium (IV)},
  title={Risk-aware Motion Planning for Autonomous Vehicles with Safety Specifications},
  year={2021},
  volume={},
  number={},
  pages={1016-1023},
  doi={10.1109/IV48863.2021.9575928}}

@ARTICLE{Yang:2026:Risk,
  author={Yang, Xuru and Zhao, Yuqiao and Hu, Yunze and Yang, Zongru and Zhu, Pingping and Sun, Ying and Liu, Chang},
  journal={IEEE Transactions on Automation Science and Engineering},
  title={Risk-Aware and Scalable Hierarchical Motion Planning for Large-Scale Robotic Swarms via {CVaR}-Constrained {MPC}},
  year={2026},
  volume={23},
  number={},
  pages={1683-1700},
  doi={10.1109/TASE.2025.3647635}
  }

@book{Royden:2023:Real,
author = {Halsey L. Royden and Patrick Fitzpatrick},
address = {Hoboken, NJ},
booktitle = {Real analysis},
edition = {5th},
publisher = {Pearson},
title = {Real analysis},
year = {2023},
}

@article{Isermann:1982:Linear,
  title={Linear lexicographic optimization},
  author={Heinz Isermann},
  journal={Operations-Research-Spektrum},
  year={1982},
  volume={4},
  pages={223-228},
}

@book{Miettinen:1999:Nonlinear,
  author    = {Kaisa Miettinen},
  title     = {Nonlinear Multiobjective Optimization},
  series     = {International Series in Operations Research \& Management Science},
  volume      = {12},
  publisher   = {Springer},
  address     = {New York, NY},
  year        = {1999},
  isbn        = {9780792382782},
  doi         = {10.1007/978-1-4615-5563-6}
}

@book{Ehrgott:2005:Multicriteria,
  author    = {Matthias Ehrgott},
  title     = {Multicriteria Optimization},
  edition    = {2},
  publisher  = {Springer},
  address    = {Berlin, Heidelberg},
  year       = {2005},
  isbn       = {9783540213987},
  doi        = {10.1007/3-540-27659-9}
}

@article{Piyavskii:1972:Algorithm,
  author  = {S. A. Piyavskii},
  title   = {An algorithm for finding the absolute extremum of a function},
  journal = {USSR Computational Mathematics and Mathematical Physics},
  volume  = {12},
  number = {4},
  pages = {57--67},
  year = {1972},
  doi = {https://doi.org/10.1016/0041-5553(72)90115-2}
}

@article{Shubert:1972:Sequential,
author = {Shubert, Bruno O.},
title = {A Sequential Method Seeking the Global Maximum of a Function},
journal = {SIAM Journal on Numerical Analysis},
volume = {9},
number = {3},
pages = {379-388},
year = {1972},
doi = {10.1137/0709036}
}

@book{Horst:1996:Global,
    title={Global Optimization: Deterministic Approaches},
    author={Horst, R. and Tuy, H.},
    year={1996},
    edition={3},
    publisher={Springer}
}

@article{Wood:1996:LipschitzEstimation,
  author  = {Wood, G. R. and Zhang, B. P.},
  title   = {Estimation of the {L}ipschitz constant of a function},
  journal = {Journal of Global Optimization},
  volume  = {8},
  pages   = {91-103},
  year    = {1996},
  doi     = {10.1007/BF00229304}
}

@inproceedings{Malherbe:2017:LIPO,
author = {Malherbe, Cedric and Vayatis, Nicolas},
title = {Global optimization of {L}ipschitz functions},
year = {2017},
publisher = {JMLR.org},
booktitle = {Proceedings of the 34th International Conference on Machine Learning - Volume 70},
pages = {2314-2323},
numpages = {10},
location = {Sydney, NSW, Australia},
series = {ICML'17}
}

@book{Farin:2001:Curves,
author = {Farin, Gerald},
title = {Curves and surfaces for CAGD: a practical guide},
year = {2001},
publisher = {Morgan Kaufmann Publishers Inc.},
address = {San Francisco, CA, USA},
edition = {5th}
}

@article{Jones:1993:Lipschitzian,
  author  = {Jones, D. R. and Perttunen, C. D. and Stuckman, B. E.},
  title   = {Lipschitzian optimization without the {L}ipschitz constant},
  journal = {Journal of Optimization Theory and Applications},
  volume  = {79},
  pages   = {157-181},
  year    = {1993},
  doi     = {10.1007/BF00941892}
}

@InProceedings{Sui:2015:Safe,
  title = 	 {Safe Exploration for Optimization with {G}aussian Processes},
  author = 	 {Sui, Yanan and Gotovos, Alkis and Burdick, Joel and Krause, Andreas},
  booktitle = 	 {Proceedings of the 32nd International Conference on Machine Learning},
  pages = 	 {997--1005},
  year = 	 {2015},
  editor = 	 {Bach, Francis and Blei, David},
  volume = 	 {37},
  series = 	 {Proceedings of Machine Learning Research},
  address = 	 {Lille, France},
  month = 	 {07--09 Jul},
  publisher =    {PMLR},
  url = 	 {https://proceedings.mlr.press/v37/sui15.html}
}

\end{document}